\documentclass[lettersize,journal]{IEEEtran}
\usepackage[linesnumbered,ruled,vlined]{algorithm2e}
\usepackage{array}
\usepackage[caption=false,font=normalsize,labelfont=sf,textfont=sf]{subfig}
\usepackage{textcomp}
\usepackage{stfloats}
\usepackage{url}
\usepackage{verbatim}
\usepackage{graphicx}
\usepackage{cite}
\usepackage{mathrsfs}
\usepackage{amsmath,amsthm,amssymb,amsfonts,mathtools}
\usepackage{algorithmic}
\usepackage{adjustbox,tabularx,threeparttable,booktabs}
\usepackage{multirow,multicol,rotating}
\usepackage{hyperref}
\usepackage[dvipsnames]{xcolor}
\usepackage[nospread]{flushend}
\usepackage{etoolbox}
\usepackage{enumitem}
\usepackage[font=small,skip=1pt]{caption}

\usepackage{soul}

\makeatletter
\patchcmd{\@makecaption}
  {\scshape}
  {}
  {}
  {}
\makeatother
\newcommand{\rot}[1]{\begin{turn}{90}#1\enspace\end{turn}}

\newtheorem{theorem}{Theorem}
\allowdisplaybreaks

\begin{document}

\title{Adaptive Probabilistic Planning for the Uncertain and Dynamic Orienteering Problem}

\author{Qiuchen Qian, Yanran Wang, and David Boyle,~\IEEEmembership{Member,~IEEE}
\thanks{The authors are with the Dyson School of Design Engineering, Imperial College London, U.K. (e-mail: qiuchen.qian19@imperial.ac.uk; yanran.wang20@imperial.ac.uk; david.boyle@imperial.ac.uk).}}

\markboth{}%
{Qian \MakeLowercase{\textit{et al.}}: Adaptive Probabilistic Planning for UDOP}


\maketitle

\begin{abstract}
The Orienteering Problem (OP) is a well-studied routing problem that has been extended to incorporate uncertainties, reflecting stochastic or dynamic travel costs, prize-collection costs, and prizes. Existing approaches may, however, be inefficient in real-world applications due to insufficient modeling knowledge and initially unknowable parameters in online scenarios. Thus, we propose the Uncertain and Dynamic Orienteering Problem (UDOP), modeling travel costs as distributions with unknown and time-variant parameters. UDOP also associates uncertain travel costs with dynamic prizes and prize-collection costs for its objective and budget constraints. To address UDOP, we develop an \textbf{AD}aptive \textbf{A}pproach for \textbf{P}robabilistic pa\textbf{T}hs, ADAPT, iteratively performing `execution' and `online planning' based on an initial `offline' solution. The execution phase updates the system status and records online cost observations. The online planner employs a Bayesian approach to adaptively estimate power consumption and optimize path sequence based on safety beliefs. We evaluate ADAPT in a practical Unmanned Aerial Vehicle (UAV) charging scheduling problem for Wireless Rechargeable Sensor Networks. The UAV must optimize its path to recharge sensor nodes efficiently while managing its energy under uncertain conditions. ADAPT maintains comparable solution quality and computation time while offering superior robustness. Extensive simulations show that ADAPT achieves a $100\%$ Mission Success Rate (MSR) across all tested scenarios, outperforming comparable heuristic-based and frequentist approaches that fail up to $70\%$ (under challenging conditions) and averaging $67\%$ MSR, respectively. This work advances the field of OP with uncertainties, offering a reliable and efficient approach for real-world applications in uncertain and dynamic environments.
\end{abstract}


\section{Introduction} \label{sec:introduction}
\IEEEPARstart{O}{}rienteering Problem (OP) is influential in many real-world applications due to its flexibility and resource constraints \cite{gunawan2016orienteering}. OP aims to determine the most efficient path that initiates from a start depot and returns to an end depot, maximizing the collected prize without violating the budget constraint. Recent research effectively extends the classic OP by introducing dynamic and stochastic attributes \cite{wang2019research, angelelli2021dpop, thayer2021adaptive}. However, a notable research gap remains concerning real-world travel costs and their potential impacts on collectible prizes and prize-collection costs. Thus, we propose a novel model, the Uncertain and Dynamic Orienteering Problem (UDOP), to further approximate real-world scenarios. Unlike existing models that often assume known probability distributions for uncertain elements, UDOP considers edge costs following distributions with unknown time-variant parameters. UDOP also considers the interrelation between edge costs, collectible prizes, and prize-collection costs. For instance, in emerging Internet of Things (IoT) contexts like urban last-mile delivery \cite{rios2021recentVRP}, travel time may vary stochastically and change systematically due to traffic congestion at different times of the day. The delayed parcel may reduce customer satisfaction (prize). The interrelation becomes interesting when edge and prize-collection costs share the same unit. Another example is the Charging Scheduling Problem (CSP) for Unmanned Aerial Vehicles (UAVs) servicing a Wireless Rechargeable Sensor Network (WRSN) \cite{li2022uavnetwork, qian2021optimal}. The energy expended during travel directly impacts the UAV's capability for recharging sensor nodes, both constrained by the residual energy budget.

In this work, we study the characteristics of UDOP in the context of the CSP for UAV-assisted WRSNs. Here, the uncertainty arises from variable factors that can cause continuous fluctuations in energy costs during UAV flights. These real-world error sources are generally unpredictable (e.g., moving obstacles, wind gusts, and turbulence). Comprehensively accounting for all unforeseen factors in mission planning is difficult due to the coupling between global mission planning and local trajectory planning (as in \cite{wang2024constrained}). This complexity may lead to suboptimal paths and uncontrollable computation time. Furthermore, effectively incorporating these factors into energy cost estimation requires additional sensor hardware support (e.g., anemometers \cite{abichandani2020wind}), specialized modeling knowledge (e.g., aerodynamics \cite{rodrigues2022drone} and battery behavior~\cite{chen2018case}). While strategies for efficiently addressing practical CSPs have been widely discussed in the literature, such as the two-stage strategy \cite{shi2024twostage}, three-dimensional charging schedule \cite{lin2023near}, and joint trajectory and scheduling optimization \cite{liu2022joint}, few researchers address the balance between solution efficiency and mission safety. In our CSP context, a safety guarantee represents that following a charging plan, the UAV can return to its end depot with sufficient energy, avoiding system failure or emergency actions like forced landing during the mission.

Thus, accurately estimating UAV energy cost is essential to guarantee mission efficiency and safety. In contrast to most literature that employs deterministic and static power consumption models for UAV scheduling \cite{pasha2022review}, we introduce adaptive probabilistic planning to continuously calibrate energy cost estimation and assess the performance of generated solutions. Our key contributions can be summarized as follows:
\begin{itemize}[leftmargin=*]
    \item We propose a new Uncertain and Dynamic Orienteering Problem (UDOP) to address real-world uncertainties. UDOP aims to identify an optimal path that maximizes the collected prize without violating budget constraints. In UDOP, edge costs follow distributions with dynamic and unknown parameters. Variations of edge costs can affect collectible prizes and prize-collection costs, bringing additional challenges to objective optimization and budget constraints. We formulate a practical UDOP with a detailed Charging Scheduling Problem (CSP) that employs a UAV to recharge sensor nodes in uncertain and dynamic environments.
    \item We propose an \textbf{AD}aptive \textbf{A}pproach for \textbf{P}robabilistic pa\textbf{T}hs, ADAPT, to address UDOP in the CSP context. ADAPT reduces the need for extra sensors and modeling knowledge requirements, enabling robust and efficient online adjustments to the UAV's path during the mission. ADAPT comprises three phases: offline planning for initial path generation, execution phase for updating the WRSN's and the UAV's status while observing actual power consumption, and online planning for updating travel costs and re-planning.
    \item The online planner incorporates a Bayesian approach to estimate the UAV's average power consumption during flight. We provide a detailed analysis of ADAPT and the Bayesian approach in Section \ref{sec:bayesian-performance}. Our empirical findings show that ADAPT can achieve a $100\%$ mission success rate with comparable solution quality and computation time across all tested scenarios, while alternative approaches have unstable performance under challenging conditions.
\end{itemize}

\section{Related Work} \label{sec:literature}
We provide an overview of research on OPs with uncertain features, focusing on approaches to solving these problems. We then consider practical strategies to address the CSP.

The OPs with uncertain attributes have recently gained attention due to their ability to model real-world uncertainties in routing problems. Gunawan et al. comprehensively review stochastic and dynamic OP variants and their associated solution approaches \cite{gunawan2016orienteering}. An interesting variant is the OP with Stochastic Travel and Service times (OPSTS) that considers uncertainty in edge travel and node service costs \cite{campbell2011orienteering}. The authors employ dynamic programming to precisely solve three special cases of the OPSTS (e.g., identical distributions for travel and service time). Angelelli et al. introduced the Dynamic and Probabilistic OP, incorporating visitation probabilities and time window constraints \cite{angelelli2021dpop}. They develop various heuristics, including static approximation, greedy methods, and Sample Average Approximation with Monte Carlo sampling. The Dynamic Stochastic OP (DSOP) assumes the travel time distributions to be discrete distributions related to the agent's arrival time \cite{lau2012DSOP}. A branch-and-bound algorithm with local search operators is applied to solve the DSOP. However, a common limitation of existing models is their reliance on a priori known probability distributions for uncertain elements. This assumption proves problematic for our UDOP as distribution parameters are unknown and evolve dynamically. Consequently, the applicability of current methods to the UDOP is constrained, necessitating new approaches to address more complex and realistic scenarios.

Considering the CSP, the primary focus is typically UAV energy management. Existing research mainly involves prolonging UAV endurance through mobile utility vehicles \cite{liu2024dynamic}, charging UAVs \cite{xue2023cuav}, and static charging stations \cite{yang2022disaster}. However, optimizing mission efficiency and safety from a planning perspective remains underexplored. For instance, Wang et al. present a framework that considers vehicle movement costs and capacity constraints \cite{wang2015mobile}. While they address the vehicle energy dynamics, their assumption of constant depletion rates may not capture real-world variations. Suman et al. propose a radio frequency energy transfer scheme considering stochastic charging efficiency due to path loss and RF-to-DC conversion \cite{suman2019uav}. However, UAV energy constraints during motion and charging operations are not adequately involved, which may impact mission safety. Evers et al. address robust UAV mission planning using uncertainty sets for edge costs and node prizes \cite{evers2014robust}. Their Robust OP (ROP) can then be optimally solved using CPLEX by selecting appropriate uncertainty intervals. ROP underscores the importance of accurately modeling the service agent's cost paradigm in real-world applications. However, precise energy cost estimation remains challenging in UAV mission planning. Recent research has focused on incorporating wind dynamics \cite{zhang2021energy, qian2022practical} and developing data-driven models \cite{rodrigues2022drone, she2020uav-power}. A notable example is the studies on the DJI M100 drone \cite{dji2016m100}. Alyassi et al. proposed a linear regression model (Reg-Model-A) requiring wind dynamics, UAV ground speed, and acceleration data \cite{alyassi2022autonomous}. In contrast, the linear regression model by Rodrigues et al. (Reg-Model-R) estimates average power consumption using the UAV's hover power \cite{rodrigues2022drone}. Our comparison of these models using realistic M100 flight data \cite{rodrigues2021flight} (see Fig.~\ref{fig:power-log}) shows that Reg-Model-R exhibits lower error while Reg-Model-A tracks the power variation better. Reg-Model-R is more practical in scenarios where accurate wind speed and acceleration predictions are infeasible, but it may underperform with insufficient empirical data or when applied to smaller, wind-sensitive UAVs. We extend this model by incorporating a Bayesian approach to adaptively update original distributions (see Sections \ref{sec:system-prior} and \ref{sec:system-online} for more details). 

\begin{figure}[t]
    \hspace*{-0.3cm}
    \centering
    \includegraphics[width=0.5\textwidth]{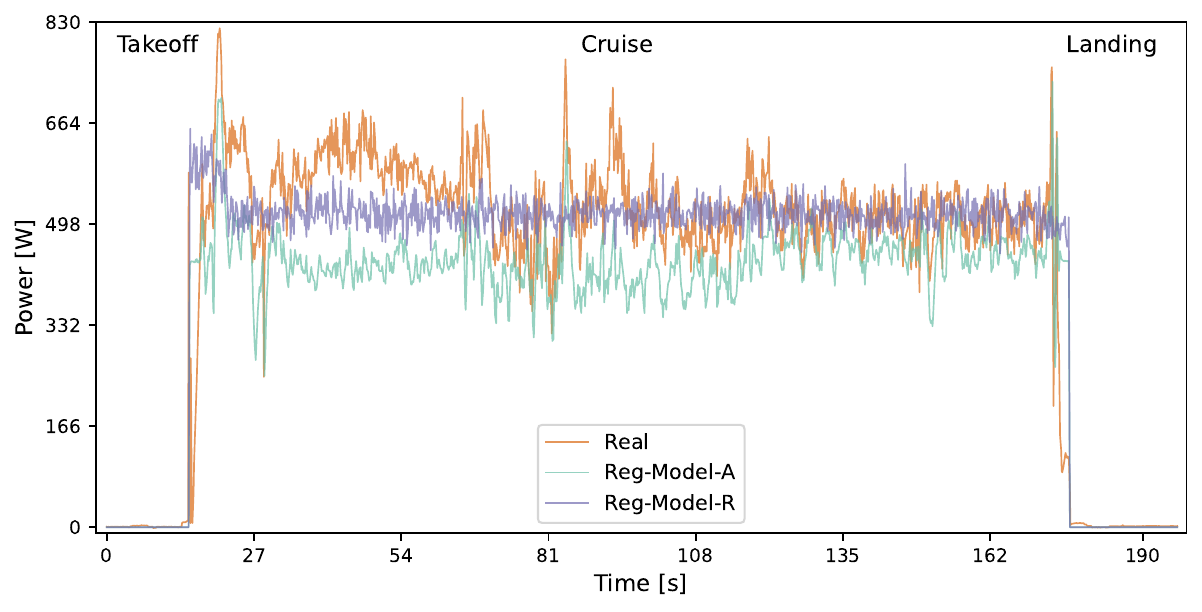}
    \caption{Linear regression models, i.e., Reg-Model-R \cite{rodrigues2022drone} and Reg-Model-A \cite{alyassi2022autonomous}, for estimating real-time power of DJI M100 \cite{dji2016m100}.}
    \label{fig:power-log}
\end{figure}

\section{Problem Formulation} \label{sec:problem}
\subsection{UDOP formulation}
Let $\mathbf{G} = \{ \mathbf{V}, \mathbf{E} \}$ be a complete graph with a set of $N$ target nodes $\mathbf{V} = \{ v_1, ..., v_n\}$ and corresponding edge set $\mathbf{E} = \{ e_{ij}, ... \}$. Each node $v_i$ is characterized by its 3D coordinate $(x_i, y_i, z_i)$ and a time-dependent prize $\mathcal{P}_{f_1}(v_i \:|\: t), \: t \in \mathbb{R}^{\geq 0}$. UDOP introduces two distinct cost functions, $\mathcal{C}_{f_2}\big(v_i, v_j \:|\: t\big)$ for the travel cost between nodes $v_i$ and $v_j$ when departing from $v_i$ at time $t$, and $\mathcal{C}_{f_3}\big(v_i \:|\: t\big)$ for the prize-collection cost at node $v_i$ when collection begins at time $t$. Here, $f_1$, $f_2$, and $f_3$ are `nominal' functions that can take any continuous form. We denote the start and the end depots by $v_0$ and $v_{N+1}$, with $\mathcal{P}(v_0) = \mathcal{P}(v_{N+1}) = 0$ and $\mathcal{C}_{f_3}\big(v_0 \:|\: t\big) = \mathcal{C}_{f_3}\big(v_{N+1} \:|\: t\big) = 0$ respectively. A feasible UDOP path begins at $v_0$, collects as much prize as possible, and ends at $v_{N+1}$, subject to a given budget constraint $\mathcal{B}$. Using binary variable $X_{ij} \in \{ 0, 1 \}$ to determine node visitation and continuous time variable $t$, we formulate UDOP as follows:
\begin{subequations}
\begin{flalign}
    \label{con:sdop-obj}
    & \textbf{(UDOP)} \quad \max\sum_{i=0}^N \sum_{j=1}^{N} \mathcal{P}_{f_1}(v_j \:|\: t_k) \: X_{ij},\quad t_k\in\mathbb{R}^{\geq0}&\\
    \label{con:sdop-start-end}
    & \textbf{s.t.} \quad \sum_{j=1}^{N+1} X_{0\: j} = \sum_{i=0}^N X_{i\: N+1} = 1&\\
    \label{con:sdop-one-visit}
    & \sum_{i=1}^{N} X_{ik} = \sum_{j=1}^{N} X_{kj} \leq 1, \quad k = 2, ..., N&\\
    \label{con:sdop-subtour}
    & \sum_{v_i \in\: \mathbf{S}}\: \sum_{v_j \in\: \mathbf{S}} X_{ij} \leq |\mathbf{S}| - 1, \quad \forall\: \mathbf{S} \subset \mathbf{V},\; |\mathbf{S}| \geq 3&\\
    \label{con:sdop-budget}
    &\begin{aligned}
        \sum_{i=0}^{N} \sum_{j=1}^{N+1} & \mathcal{C}_{f_2}\big( v_i, v_j \:|\: t_k \big) \cdot X_{ij} &\\
        + \sum_{i=0}^{N} & \sum_{j=1}^{N} \mathcal{C}_{f_3}\big( v_j \:|\: t_l \big) \cdot X_{ij} \:\leq\:\mathcal{B}, \quad 0 \leq t_k < t_l &
    \end{aligned}&
\end{flalign}
\end{subequations}
The objective function \eqref{con:sdop-obj} maximizes the collected prizes. Constraint \eqref{con:sdop-start-end} ensures the path starts at $v_0$ and ends at $v_{N+1}$. Constraint \eqref{con:sdop-one-visit} maintains path connectivity and restricts each target node to at most one visit. Constraint \eqref{con:sdop-subtour} prevents subtours, ensuring a single continuous path. Constraint \eqref{con:sdop-budget} stipulates the total path cost, comprising prize collection and travel costs, must not exceed the given budget $\mathcal{B}$.

\subsection{CSP formulation}
The CSP assigns a single UAV to recharge the maximum energy while ensuring the UAV's safe return to the end depot under uncertain environments. To tackle the UDOP within the CSP scenario, understanding the analytical form of $f_1, f_2$ and $f_3$ is essential. In principle, all three functions can be stochastic and initially unknowable, but we make several assumptions as below to reasonably reduce the problem's complexity based on established research. The recharging process for sensor nodes involves a DC-DC converter, inverter, inductive link, rectifier, and constant current (CC) charger (as illustrated in \cite{arteaga2023high}, Figure 17). Drawing from their experimental results, we simplify their Inductive Power Transfer (IPT) process using fixed efficiencies: $\eta_{\text{IPT}}$ for the IPT link\footnote[1]{$\eta_{\text{IPT}}$ can be statistically characterized by IPT environment distributions \cite{arteaga2018probability}, although such modeling is not required to demonstrate ADAPT's utility.} and $\eta_{\text{CC}}$ for the CC charger. The charger operates within a 20-42 V range at CC, charging a $\mathrm{C}=10$ F supercapacitor bank with an average current of $\bar{\mathrm{I}}_{\text{CC}} = 0.825$ A, producing a $[0, 6.82]$ kJ node prize range. For instance, with a CC charger voltage at 30 V, the UAV's energy consumption and charging time are given by: $\mathrm{E}_{\text{IPT}} = 0.5 \: \mathrm{C} \: (\mathrm{V}_{\max}^2 - 30^2) / \eta_{\text{IPT}}$ and $t_{\text{IPT}} = 10 \cdot (\mathrm{V}_{\max} - 30) / \bar{\mathrm{I}}_{\text{CC}} / \eta_{\text{CC}}$. We assume all sensor nodes are identical to those described in \cite{polonelli2020flexible} and operate at a constant sampling frequency. This leads to a uniform energy depletion rate, $R_{\text{SN}}$, which linearly increases energy and time requirements for recharging sensor nodes. Based on these assumptions, we define $f_1$, $f_2$ and $f_3$ as follows:
\begin{subequations}
\begin{flalign}
    \label{con:csp-prize}
    & \mathcal{P}_{f_1} (v_i \:|\: t) = \frac{1}{2}\: \mathrm{C}\:\big(\mathrm{V}_{\max}^2 - \mathrm{V}(v_i \:|\: t = 0) ^2 \big) + R_{\text{SN}}\: t&\\
    \label{con:csp-travel-cost}
    & \mathcal{C}_{f_2} (v_i, v_j \:|\: t) = \frac{\mathrm{\bar{P}}^*_{\text{tk}}\: (H - z_i)}{\mathrm{v}_{\text{tk}}} + \frac{\mathrm{\bar{P}}^*_{\text{cr}}\: d(v_i, v_j)}{\mathrm{v}_{\text{cr}}} + \frac{\mathrm{\bar{P}}^*_{\text{ld}}\: (H - z_j)}{\mathrm{v}_{\text{ld}}}&\\
    \label{con:csp-charge-cost}
    & \mathcal{C}_{f_3} (v_i \:|\: t) = \frac{\mathcal{P}_{f_1} (v_i \:|\: t)}{\eta_{\text{IPT}}}&
\end{flalign}
\end{subequations}
Equation \eqref{con:csp-prize} defines a sensor node's chargeable energy as the difference between its maximum energy capacity and current energy level. Equation \eqref{con:csp-travel-cost} stipulates the UAV's energy consumption during travel, accounting for three distinct flight regimes: takeoff, cruise, and landing. $\mathrm{\bar{P}}^*$ denotes the actual average power consumption, following a normal distribution with unknown mean and standard deviation (SD), e.g., $\mathrm{\bar{P}}^*_{\text{tk}}(t_1, t_2) \sim \boldsymbol{\mathsf{N}}(\mu_{\text{tk}}, \sigma_{\text{tk}} \:|\: t_1, t_2)$. $H$ refers to the cruise altitude, $d(v_i, v_j)$ is the Euclidean distance between two nodes, and $\mathrm{v}$ denotes the average speed in each regime. Equation \eqref{con:csp-charge-cost} quantifies the energy needed to recharge a sensor node as the ratio of chargeable energy to the IPT link efficiency.

\section{System Design} \label{sec:system}
\begin{figure}[t]
    \hspace*{-0.15cm}
    \centering
    \includegraphics[width=0.5\textwidth]{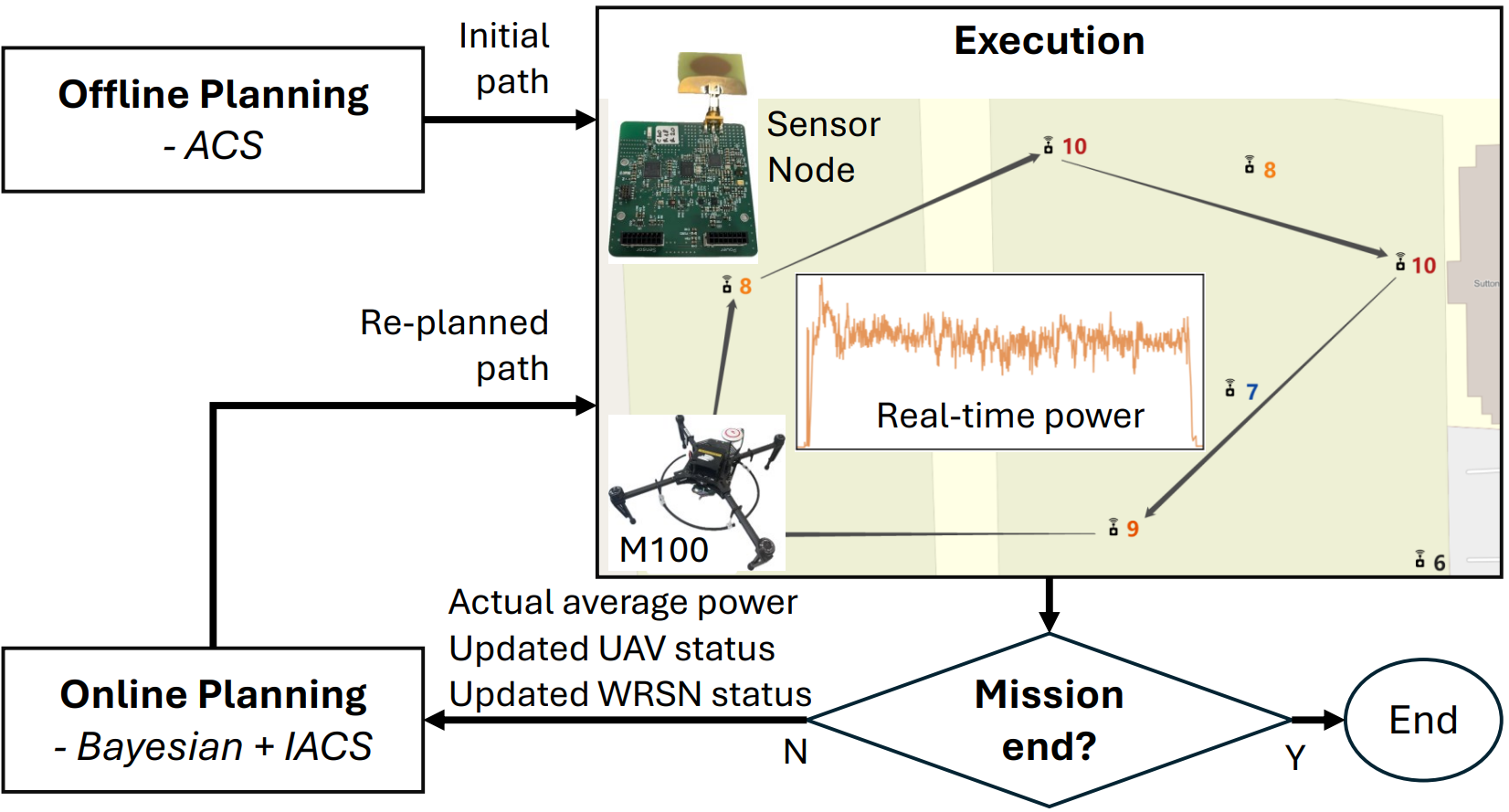}
    \caption{ADAPT framework. The UAV follows an initial path (generated offline), sequentially servicing sensor nodes. During flight, the UAV continuously logs power consumption, which informs subsequent planning triggered upon completing the task at each node. The iterative process of execution and online planning phases continues until the UAV returns to the end depot once all target nodes are recharged, or the residual energy is insufficient to continue.}
    \label{fig:flow-chart}
\end{figure}

\begin{table}[t]
\centering
\begin{threeparttable}[t]
\caption{Main parameters in ADAPT.} \label{tab:ALG-PARAMS}
\setlength\tabcolsep{0pt}
\begin{tabular*}{0.49\textwidth}{@{\extracolsep{\fill}} ll}
\toprule
Parameters & Definition \\
\midrule
$\bar{\mathrm{P}}$, $\bar{\mathrm{P}}^*$ & Estimated and actual average power during UAV flight (W). \\
$m$, $\rho$ & UAV weight (kg) and air density (kg/m$^3$). \\
$\mu_{\bar{\mathrm{P}}}$, $\mu_{\bar{\mathrm{P}}^*}$ & Mean of distributions for $\bar{\mathrm{P}}$ and $\bar{\mathrm{P}}^*$. \\ 
$\sigma_{\bar{\mathrm{P}}}$, $\sigma_{\bar{\mathrm{P}}^*}$ & Standard deviation of distributions for $\bar{\mathrm{P}}$, $\bar{\mathrm{P}}^*$. \\ 
$\Delta\mu$, $\Delta\sigma$ & The shift to original $\mu$ and $\sigma$ (\%). \\
$\Theta$ & Using CDF probabilities from the distribution as costs to \\ 
 & derive a feasible solution is having a safety belief $\Theta$ (\%). \\ 
$\mathcal{P}$, $\mathcal{C}$ & Prize (kJ) and Cost (kJ). \\
\bottomrule
\end{tabular*}
\end{threeparttable}
\end{table}

This section details the framework design of ADAPT for the CSP. ADAPT aims to provide outer loop control, dynamically adjusting the visitation sequence in response to environmental changes. It comprises three phases: offline planning, execution, and online planning (shown in Fig.~\ref{fig:flow-chart}). The UAV mission initiates with a pre-computed offline path. During the execution phase, it updates the observed travel costs (i.e., via continuous power consumption measurement), node prizes, and prize-collection costs for all unvisited sensor nodes. The execution and online planning phases are iteratively conducted until meeting termination criteria: successful charging of all target nodes or mandatory return-to-home due to insufficient energy to continue. Table~\ref{tab:ALG-PARAMS} includes main parameters used in ADAPT.

\subsection{Prior knowledge for edge costs} \label{sec:system-prior}
The complex UDOP can be approximated to a classic OP with static and deterministic prizes and costs under the CSP scenario. Because voltages of sensor nodes remain nearly the same over short intervals, we can assume static prizes and prize-collection costs during planning but re-plan regularly after each recharging operation. ADAPT decouples those dynamics from the planning phases by re-estimating chargeable energy and recharging costs in the execution phase. Furthermore, for UAVs operating at 25-100 m altitudes and ground speeds of 4-12 m/s, wind conditions varying from approximately 3.89 to 8.23 m/s have little impact on average energy consumption \cite{rodrigues2022drone}. Consequently, the induced power at hover in no-wind conditions serves as an adequate estimate for average power consumption $\bar{\mathrm{P}}$ during flight \cite{rodrigues2022drone}:
\begin{equation}
    \bar{\mathrm{P}} = b_1 \sqrt{\frac{m^3}{\rho}} + b_0
\end{equation}
where $b_1$ and $b_0$ are coefficients derived from linear regression analysis between induced power and actual average power. Constants $m$ and $\rho$ denote the UAV weight and air density, respectively. Table~\ref{tab:reg-model-coef} presents the trained coefficients for three distinct UAV regimes. For details, we refer the reader to supplementary files of \cite{rodrigues2022drone}.
\begin{theorem}
    The estimated average power consumption can be modeled as a normal distribution with mean $\mu_{\mathrm{\bar{P}}}=\mu(b_1) \sqrt{\frac{m^3}{\rho}} + \mu(b_0)$ and variance $\sigma^2_{\mathrm{\bar{P}}^*} = \sigma^2(b_1) \frac{m^3}{\rho} + \sigma^2(b_0)$.    
\end{theorem}
\begin{proof}
    See Appendix A.
\end{proof}
\noindent
Leveraging constraint \eqref{con:csp-travel-cost}, the average power consumption can be sampled at discrete time steps, with the summation of these samples providing an estimate of the travel cost. Thus, a static and deterministic graph suffices as input for planners.

\begin{table}[t]
\centering
\begin{threeparttable}[t]
\caption{Model coefficient $\pm$ bootstrap standard error \cite{rodrigues2022drone}.} \label{tab:reg-model-coef}
\setlength\tabcolsep{0pt}
\begin{tabular*}{0.485\textwidth}{@{\extracolsep{\fill}} cccc}
\toprule
Coefficient & Takeoff & Cruise & Landing \\
\midrule
$b_1$ & $80.4 \pm 2.6$ & $68.9 \pm 2.0$ & $71.5 \pm 1.7$ \\
$b_0$ & $13.8 \pm 18.9$ & $16.8 \pm 15.0$ & $-24.3 \pm 12.5$\\
\bottomrule
\end{tabular*}
\end{threeparttable}
\end{table}

\subsection{Offline planning phase}
The `offline' planning phase, executed once at the start depot, aims to generate a high-quality initial path for the UAV. Because designing a new algorithm for solving the static OP is beyond the scope of this study, we modify a well-established discrete metaheuristic as the main solver for offline and online planning phases. Specifically, we select the Ant Colony System (ACS) \cite{dorigo1997ant} (see Appendix B for details) due to its straightforward implementation, high adaptability, effectiveness for various instances of OP \cite{mandal2020survey}, and convergence towards the optimal solution \cite{Stutzle2002ACSproof}. It is essential to highlight that even an optimal offline solution cannot guarantee global performance for the whole mission. In practice, ACS exhibits a good balance between solution quality and computational efficiency, aligning with the dynamic nature of our online problem. To validate this, we compare ACS with an exact method implemented by Gurobi \cite{gurobi} and investigate the potential impact of different offline paths on later planning in Appendix C.

\subsection{Execution phase}
In this phase, the UAV proceeds to the next target node as determined by the offline or online planning. It continuously monitors and records the battery's real-time power output throughout the flight to calculate the actual average power consumption $\mathrm{\bar{P}}^*$. Upon completion of the recharging process at each node, the UAV updates the estimated chargeable energy $\mathcal{P}_{f_1} (v_i\:|\: t)$ and recharging cost $\mathcal{C}_{f_3} (v_i \:|\: t)$ for all unvisited sensor nodes based on the actual travel and charging time. 

\subsection{Online planning phase} \label{sec:system-online}
The online planner takes inputs of power observations, UAV status (coordinate and residual energy), and updated WRSN status (chargeable energy and recharging cost). If we assume $\bar{\mathrm{P}}^*$ follows a normal distribution with unknown mean $\mu_{\bar{\mathrm{P}}^*}$ and variance $\sigma^2_{\bar{\mathrm{P}}^*}$, these parameters can be inferred with a conjugate Normal-Gamma (NG) prior distribution through Bayesian Inference (BI) \cite{FriebeMPN21, Turlapaty20}. Employing BI offers advantages to estimate $\bar{\mathrm{P}}^*$ in the CSP because it: (a) incorporates prior knowledge to make estimation effective in early mission stages; (b) quantifies uncertainty in an interval, enabling robust decision-making under variable conditions; (c) allows continuous updating of estimates as new data becomes available, making it ideal for online scenarios where power consumption fluctuates; (d) convergence to the true distribution as observations accumulate \cite{savchuk2011bayesian}. Moreover, pre-training a regression model as in \cite{rodrigues2022drone} can ease BI's limitations, e.g., the need to specify prior distributions and sensitivity to prior. To validate BI performance, we compare it to alternative methods, including frequentist and heuristic-based approaches (see Sections \ref{sec:bayesian-performance} and \ref{sec:alg-performance}).

The hyperparameters of $\boldsymbol{\mathsf{NG}}(\mu_{\bar{\mathrm{P}}^*}\: ,\: \sigma^{-2}_{\bar{\mathrm{P}}^*} \:|\: \mu_0, \kappa_0, \alpha^{\text{BI}}_0, \beta^{\text{BI}}_0)$ can be determined using $\mu_{\bar{\mathrm{P}}}$ and $\sigma_{\bar{\mathrm{P}}}$. Following \cite{murphy2007conjugate}, the posterior parameters can be updated as:
\begin{subequations}
\begin{flalign}
    \label{eq:post-updated-mu}
    &\mu_n = \frac{\kappa_0 \mu_0 + n \bar{\mathbf{x}}}{\kappa_0 + n}&\\
    \label{eq:post-updated-kappa}
    &\kappa_n = \kappa_0 + n&\\
    \label{eq:post-updated-alpha}
    &\alpha_n^{\text{BI}} = \alpha_0^{\text{BI}} + n/2&\\
    \label{eq:post-updated-beta}
    &\beta_n^{\text{BI}} = \beta_0^{\text{BI}} + \frac{1}{2} \sum_{i=1}^{n} (x_i - \bar{\mathbf{x}})^2 + \frac{\kappa_0 n (\bar{\mathbf{x}} - \mu_0)^2}{2 (\kappa_0 + n)}&
\end{flalign}
\end{subequations}
where $n$ refers to sample size, and $\bar{\mathbf{x}}$ denotes sample mean. We employ a sliding time window to omit out-of-date observations. The window length is determined by two factors: the observed average power consumption may be affected by uncommon extreme values (e.g., significant fluctuations due to wind gusts or turbulence \cite{kim2020flight}), but it should still reflect recent environmental changes. The observed data and NG prior result in a posterior predictive of a Student-t distribution with center at $\mu_n$, precision $\Lambda = \frac{\alpha_n^{\text{BI}} \kappa_n}{\beta_n^{\text{BI}} (\kappa_n + 1)}$ and degree of freedom $\nu = 2\alpha_n^{\text{BI}}$ \cite{murphy2007conjugate}. Therefore, new estimated average power consumption can be sampled from this location-scale t distribution, i.e.,
\begin{equation}
    \mathrm{\bar{P}} \sim \mu_n + \boldsymbol{\mathsf{t}}_{2\alpha_n^{\text{BI}}} \sqrt{\frac{\beta_n^{\text{BI}} (\kappa_n + 1)}{\alpha_n^{\text{BI}} \kappa_n}}
\end{equation}
The posterior's Cumulative Distribution Function (CDF) is used to derive various power consumption levels for takeoff, cruise, and landing, forming potential edge costs for travel. A path planned using average power $\bar{\mathrm{P}}$, obtained with CDF probability $\Theta = p(X \leq \bar{\mathrm{P}})$, is defined as having a \textit{safety belief} $\Theta$ to complete the mission. Thus, we reduce UDOP to a deterministic and static problem with a given $\Theta$ value.

The internal solver is a modified version of the Inherited ACS (IACS) proposed by \cite{qian2024ceopn} (see Appendix B). The inheritance mechanism, designed to advance convergence, naturally aligns with the iterative process of execution and online planning. Specifically, the best path from the preceding planning iteration is a superior initialization for the pheromone matrix, surpassing the conventional nearest neighbor heuristic. The drop operator ensures path feasibility, eliminating low-value nodes to adhere to budget constraints. The add operator aims to improve path quality by inserting high-value feasible nodes to maximize prize collection. IACS is then applied multiple times to search candidate paths within a safety belief range of $\big[\Theta_{\min}, \Theta_{\max}\big]$. The final output is the path with the highest weighted score between safety belief and solution quality:
\begin{equation}
    \boldsymbol{S} = \arg\max_{S_i} \Big\{w_{\Theta} \frac{\Theta_i - \Theta_{\min}}{\Theta_{\max} - \Theta_{\min}} + w_{\mathcal{P}} \frac{\mathcal{P}(S_i) - \mathcal{P}_{\min}}{\mathcal{P}_{\max} - \mathcal{P}_{\min}} \Big\}
\end{equation}
\noindent
where $w_{\Theta}$ and $w_{\mathcal{P}}$ are factors to balance the weight of safety belief and prize collection.

\section{Experiments, Results and Discussion} \label{sec:experiment}
This section presents numerical results, evaluating the performance and robustness of ADAPT and benchmark approaches. We employ three test instances, denoted as \textit{California20}, \textit{California30}, and \textit{California40}, representing WRSNs randomly deployed in 1 km$^2$ area in California, with an increasing number of sensor nodes\footnote[2]{Software implementation and experimental results of ADAPT are available by link \url{https://github.com/sysal-bruce-publication/Uncertain-Dynamic-OP.git}. \label{fn-github}}. To assess the algorithms' robustness under uncertain and variable conditions, we \textbf{shift and scale} prior normal distributions with coefficients in Table~\ref{tab:reg-model-coef} to represent distributions of actual average power consumption $\mathrm{\bar{P}^*}$, which is unknown to the planner. For example, a windy scenario might be characterized by distributions with actual mean $\mu_{\mathrm{\bar{P}}^*} = 110\% \: \mu_{\bar{\mathrm{P}}}$ and SD $\sigma_{\mathrm{\bar{P}}^*} = 120\% \: \sigma_{\bar{\mathrm{P}}}$. For simplicity, we denote this adjustment as $\Delta\mu_{\mathrm{\bar{P}}^*} = 10\%, \Delta \sigma_{\mathrm{\bar{P}}^*} = 20\%$, respectively. Unless specifically stated otherwise (e.g., to highlight specific scenarios), we evaluate each approach using the stated test instances under various actual power distributions, i.e., $\Delta\mu_{\mathrm{\bar{P}}^*},\Delta\sigma_{\mathrm{\bar{P}}^*} \in \{-10, 0, 10, 20\}\%$. Though 20 individual executions are proven sufficient to examine repeatability \cite{qian2024ceopn}, we increase executions to 50 due to random sampling from distributions. We assess algorithm robustness and performance using three metrics: Mission Success Rate (MSR), actual collected prize $\mathcal{P}^*$, and actual cost $\mathcal{C}^*$.

\subsection{Parameter setting} \label{sec:param-setting}
The execution phase models the UAV's energy consumption primarily through travel and service operations within the context of CSP. For travel modeling, parameters are based on \cite{rodrigues2021flight, rodrigues2022drone, dji2016m100}. The total weight of the M100 drone is 3.93 kg, including a TB47D battery (359.64 kJ capacity) and 0.25 kg payload of induction coil and driving circuits. The air density is set to a common value $\rho=1.225$ kg/m$^3$ \cite{zhang2021energy}. The UAV flight protocol consists of takeoff (ascend to $H = 30$ m with an average speed $\mathrm{v}_{\text{tk}} = 3$ m/s), cruise (travel to the next waypoint with an average speed $\mathrm{v}_{\text{cr}} = 10$ m/s) and landing (descend to the ground with an average speed $\mathrm{v}_{\text{ld}} = 2$ m/s). The estimated average power consumption $\bar{\mathrm{P}}$ for each regime is sampled from normal distributions with coefficients stated in \cite{rodrigues2022drone}, i.e., $\bar{\mathrm{P}}_{\text{tk}} \overset{\mathrm{i.i.d.}}{\sim} \boldsymbol{\mathsf{N}}(579.75, 692.16)$, $\bar{\mathrm{P}}_{\text{cr}} \overset{\mathrm{i.i.d.}}{\sim} \boldsymbol{\mathsf{N}}(501.80, 423.20)$, $\bar{\mathrm{P}}_{\text{ld}} \overset{\mathrm{i.i.d.}}{\sim} \boldsymbol{\mathsf{N}}(479.00, 299.45)$. The service simulation models the UAV's recharging process for homogeneous sensor nodes. Following experimental results in \cite{arteaga2023high}, we set the IPT link efficiency $\eta_{\text{IPT}} = 40\%$, the CC charger efficiency $\eta_{\text{CC}} = 90\%$ and the energy depletion rate of sensor nodes $R_{\text{SN}} = 2.19 \cdot 10^{-6}$ kJ/s \cite{polonelli2020flexible}.

\begin{figure}[b]
    \hspace*{-0.4cm}
    \centering
    \includegraphics[width=0.52\textwidth]{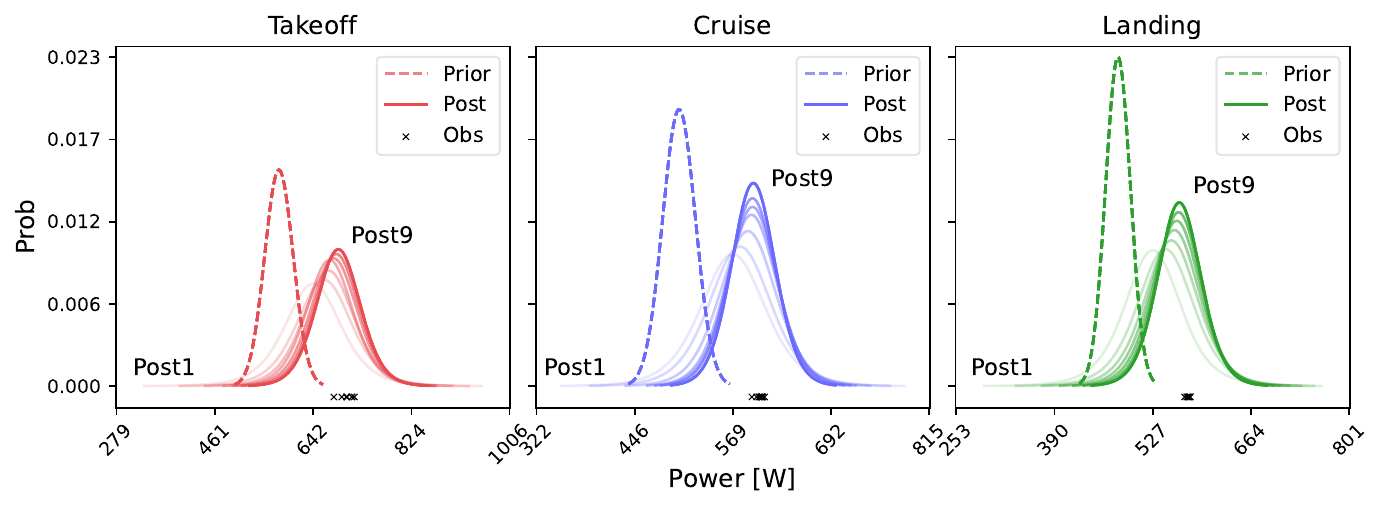}
    \caption{An example of how ADAPT updates posterior distributions using online observations. Nine re-plannings happened during this mission, moving from the most diverged distribution (Post1) to the most centralized one (Post9).}
    \label{fig:all-dists}
\end{figure}

\begin{figure*}[b] 
    \centering
    \hspace*{0.25cm}
    \includegraphics[width=0.95\textwidth]{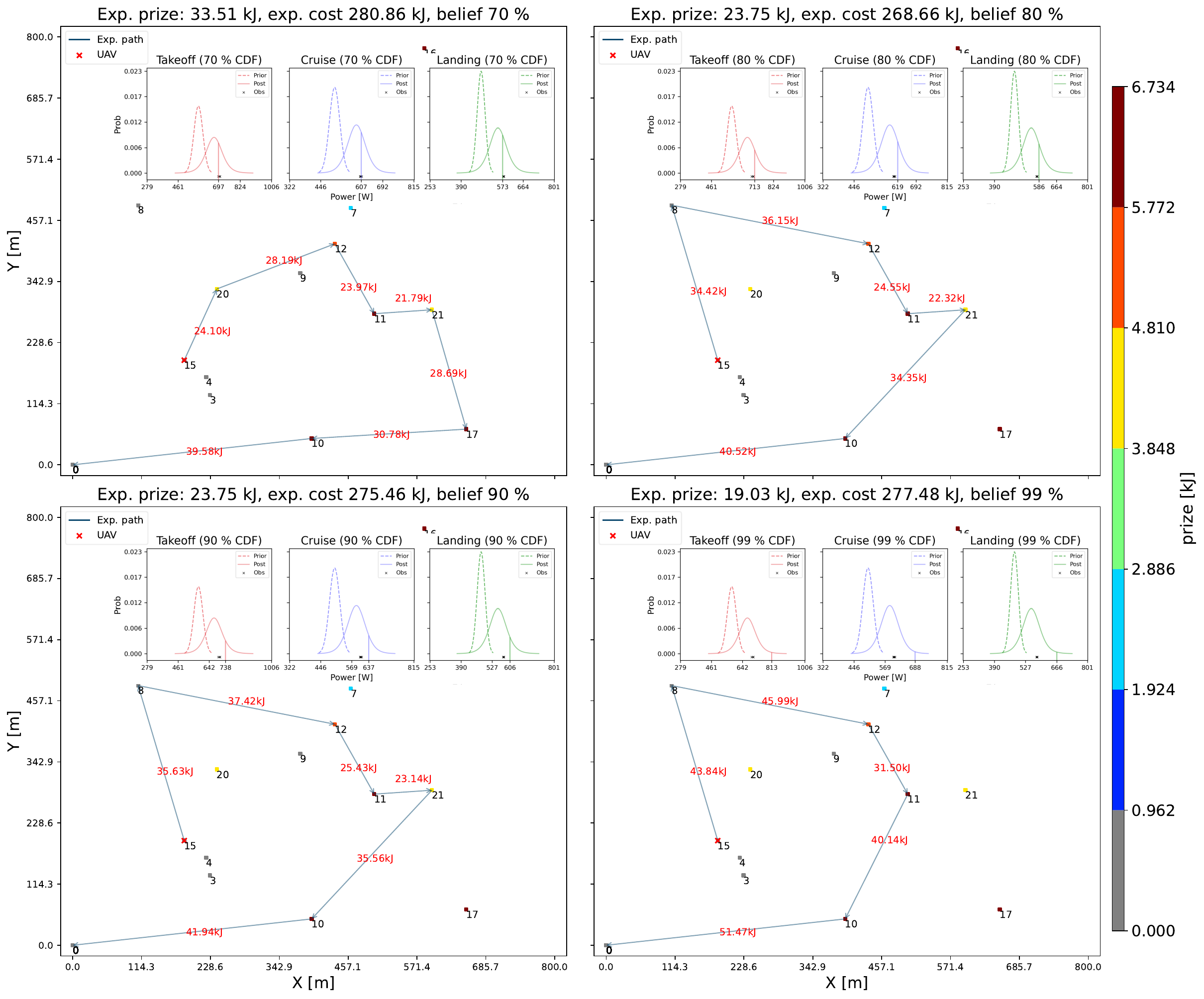}
    \caption{An example of ADAPT solving \textit{California20} with $\Delta\mu_{\mathrm{\bar{P}}^*}= 20\%$ and $\Delta\sigma_{\mathrm{\bar{P}}^*} = 0\%$. When the UAV recharged sensors 4 and 15, its residual energy (budget) is 283.08 kJ. Here we show four typical candidate paths of ADAPT with safety belief $\Theta\in \{70, 80, 90, 99\}\%$.}
    \label{fig:path}
\end{figure*}

During the execution, we assume the period of average power reading as 20 seconds and the sliding time window length as the latest 15 minutes to balance historical observation utilization and temporal sensitivity. For Bayesian Inference, we set $\alpha_0^{\text{BI}} = 2$, $\beta_0^{\text{BI}} = \sigma^2_{\bar{P}}$, $\mu_0 = \mu_{\bar{P}}$, and $\kappa_0 = 1$ to employ a weekly informative prior knowledge about the mean and variance. We assign unbiased weights to the safety belief and prize collection, i.e., $w_{\Theta} = w_{\mathcal{P}} = 50\%$ in the whole mission. The sensitivity analysis of $\Theta_{\min}$ with a range from $45\%$ to $85\%$ is presented in Appendix E. Our experimental results indicate that smaller $\Theta_{\min}$ values (e.g., $45\%$ and $55\%$) generally lead to solutions with slightly higher prizes but lower mission success rates. Therefore, we set $\Theta_{\min} = 75\%$ and $\Theta_{\max} = 99.9\%$ to prioritize mission safety with a minor sacrifice of prize collection. Later, we demonstrate that ADAPT can still outperform other benchmark approaches with $\Theta_{\min} = 75\%$ under some scenarios in Section \ref{sec:alg-performance}. 

Based on \cite{qian2024ceopn}, we set solver parameters as: Number of ants $N_{\text{ant}} = 40$, Number of iterations $N_{\text{it}} = 250$, heuristic importance factor $\beta_{\text{ACS}} = 2$, pheromone evaporation rate $\alpha_{\text{ACS}} = \rho = 0.1$. In the offline planning phase, the initial pheromone $\tau_0 = \mathcal{P}_{\text{nn}} / (\mathcal{C}_{\text{nn}} \cdot (|\mathcal{S}_{\text{nn}}| - 1))$, where $\mathcal{P}_{\text{nn}}, \mathcal{C}_{\text{nn}},$ and $|\mathcal{S}_{\text{nn}}|$ are the path prize, path cost, and path length of the solution achieved by nearest neighbor heuristic, respectively. While in the online planning phase, the initial pheromone $\tau_0 = \mathcal{P}^{S^{\text{gb}}(n_{\text{it}}-1)} / (\mathcal{C}^{S^{\text{gb}}(n_{\text{it}}-1)} \cdot (|S^{\text{gb}}(n_{\text{it}}-1)| - 1))$, where $S^{\text{gb}}(n_{\text{it}}-1)$ is the global-best solution obtained from previous iteration. We establish a minimum improvement tolerance $\epsilon_{\text{ACS}} = 10^{-4}$, which means ACS would terminate if the fitness difference is less than $\epsilon_{\text{ACS}}$ for several iterations. To balance the computation time and solution quality, we allow a maximum number of no improvements as $N_{\text{impr}} = N_{\text{ACS}} / 10 = 25$. 

\subsection{The online planning phase with a Bayesian approach} \label{sec:bayesian-performance}
We first demonstrate the performance of the online planning phase through a working example of a \textit{California20} mission. In this scenario, we set the actual power consumption mean $20\%$ higher than estimated $\Delta \mu_{\mathrm{\bar{P}}^*} = 20\%$ while keeping its SD unchanged $\Delta \sigma_{\mathrm{\bar{P}}^*} = 0\%$. Fig.~\ref{fig:all-dists} illustrates the evolution of posterior distributions (from Post1 to Post9), based on given prior distributions and continuous online observations mission, with $\Delta \mu_{\mathrm{\bar{P}}^*} = 20\%$ and $\Delta \sigma_{\mathrm{\bar{P}}^*} = 0\%$. Post1 exhibits a large SD due to insufficient samples, and the observed data significantly differs from the prior mean. As additional observations accumulate, posterior distributions demonstrate increasing centralization, reflecting ADAPT's adaptive learning process. ADAPT subsequently updates the edge cost matrix using these posterior distributions. Fig.~\ref{fig:path} demonstrates the impact of varying safety belief ($\Theta$) values (representing the confidence level in completing the mission) on solution quality for the above scenario. Paths with higher $\Theta$ values tend to be more conservative, while those with lower $\Theta$ values expect to charge more nodes. Notably, using $\Theta$ values of $80\%$ and $90\%$ yields identical solutions, suggesting this path can accommodate higher travel costs without compromising expected prize collection. This result underscores the ADAPT's capability to balance risk and reward effectively.

\begin{figure}[b]
    \hspace*{-0.35cm}
    \centering
    \includegraphics[width=0.51\textwidth]{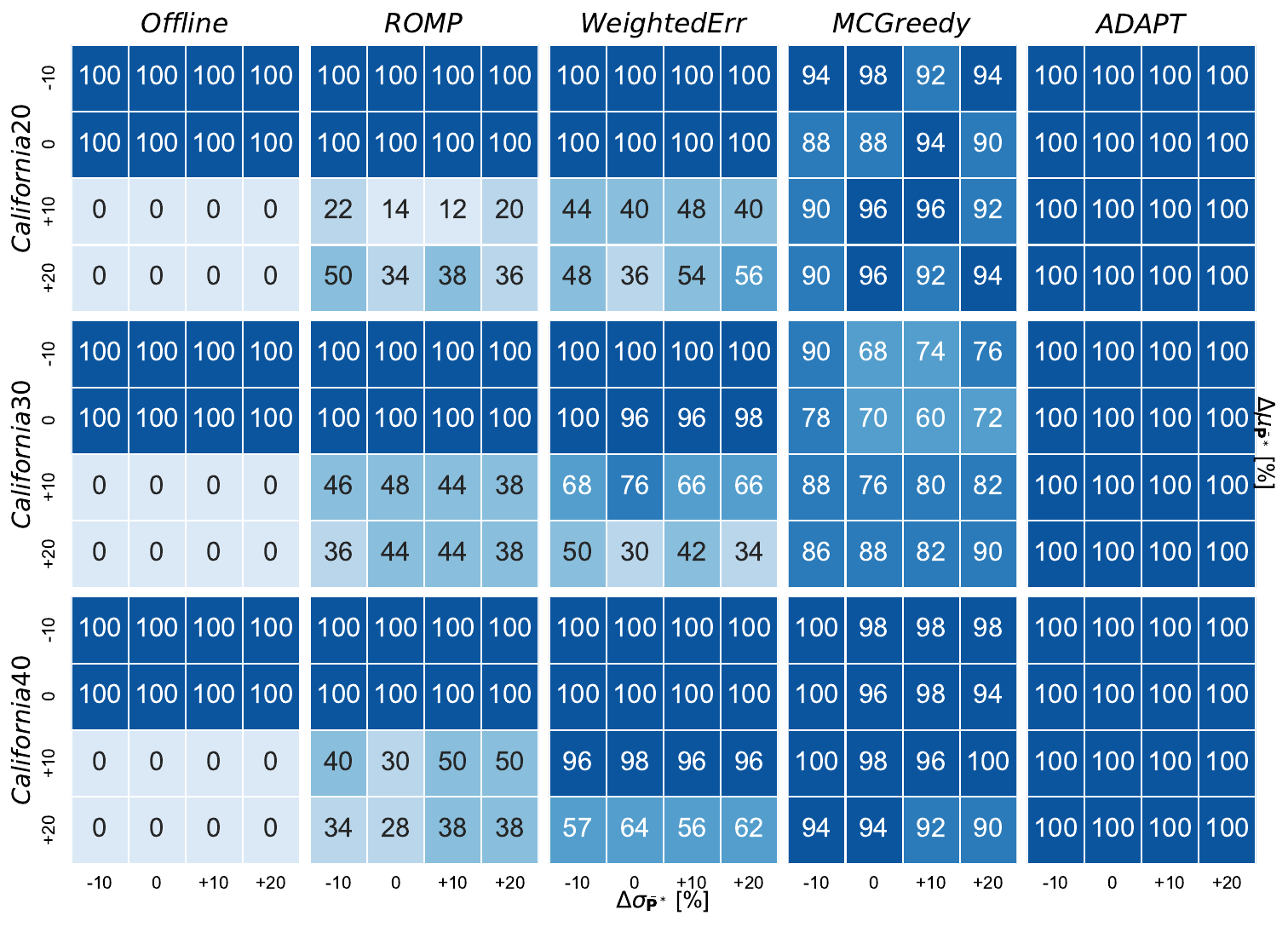}
    \caption{Mission success rate over 50 executions.}
    \label{fig:prob-of-succ}
\end{figure}

\subsection{ADAPT performance analysis of computation time, mission safety and solution quality} \label{sec:alg-performance}
We evaluate ADAPT's performance against four alternative approaches: 
\begin{itemize}[leftmargin=*]
    \item \textit{Offline} always follows the initial offline path during the whole mission. It shows the performance of a static approach that ignores new information during the mission.
    \item Rapid Online Mission Planner (\textit{ROMP} \cite{qian2022practical}) represents a simple adaptive strategy. It re-plans at each node, using prior travel costs estimated during the offline planning phase.
    \item \textit{WeightedErr} is a heuristic-based method that dynamically updates energy costs from recent observations. It calculates the weighted error ratio between estimated and actual energy costs from the most recent travel: $R_{\text{err}} = w_\text{act} (\frac{\Delta \mathrm{E}_\text{act} - \Delta \mathrm{E}_\text{est}}{\Delta \mathrm{E}_{\text{est}}} + 1) + w_{\text{est}}$. Energy costs of all feasible edges are then updated as $\mathrm{E}'(v_i, v_j) = R_{\text{err}} \cdot \mathrm{E}(v_i, v_j)$. We set fixed weights $w_\text{act} = w_\text{est} = 0.5$ to balance sensitivity to estimation errors. 
    \item Monte Carlo Greedy (\textit{MCGreedy} \cite{angelelli2021dpop}) randomly samples $N_{\text{MC}}$ power levels between the minimum and maximum observed power consumption. The final output is the candidate path with the highest occurrence frequency. We set the number of samples $N_{\text{MC}} = 100$ as in \cite{angelelli2021dpop}. As a frequentist approach, \textit{MCGreedy} compares to ADAPT regarding how uncertainty is quantified and used in decision-making.
\end{itemize}
To ensure a fair comparison, all approaches employ the same solver (i.e., IACS) in the online re-planning phase. Because all approaches can complete computation within seconds (see Appendix D), we omit execution time in the following results.

\begin{table}[b]
\centering
\begin{threeparttable}[t]
\caption{Solution quality comparison.} \label{tab:sol-quality-short}
\setlength{\tabcolsep}{0pt}
\begin{tabular*}{0.489\textwidth}{@{\extracolsep{\fill}} ccccccccccc}
\toprule
& \multicolumn{1}{c}{$\Delta\mu_{\mathrm{\bar{P}}^*}$} & \multicolumn{1}{c}{$\Delta\sigma_{\mathrm{\bar{P}}^*}$}  & \multicolumn{2}{c}{\textit{ROMP}} & \multicolumn{2}{c}{\textit{WeightedErr}} & \multicolumn{2}{c}{\textit{MCGreedy}} & \multicolumn{2}{c}{\textit{ADAPT}} \\ 
\cline{4-5} \cline{6-7} \cline{8-9} \cline{10-11}
 & (\%) & (\%) & $\mathcal{P}^*$(kJ) & $\mathcal{C}^*$(kJ) & $\mathcal{P}^*$(kJ) & $\mathcal{C}^*$(kJ) & $\mathcal{P}^*$(kJ) & $\mathcal{C}^*$(kJ) & $\mathcal{P}^*$(kJ) & $\mathcal{C}^*$(kJ) \\ \midrule
\multirow{8}{*}{\rot{\small\textit{California20}}} & 10 & -10 & \textbf{45.45} & 357.71 & 44.55 & 355.27 & 43.21 & 351.03 & 44.57 & 349.93 \\
 & 10 & 0 & \textbf{46.22} & 359.27 & 44.48 & 355.43 & 43.15 & 351.21 & 44.47 & 350.05 \\
 & 10 & 10 & \textbf{45.84} & 358.48 & 43.96 & 354.05 & 42.53 & 350.43 & 44.40 & 349.98 \\
 & 10 & 20 & \textbf{45.71} & 357.44 & 45.02 & 356.91 & 42.80 & 350.41 & 44.72 & 351.77 \\
 & 20 & -10 & 39.57 & 357.19 & \textbf{42.05} & 357.94 & 40.93 & 351.68 & 40.98 & 349.41 \\
 & 20 & 0 & 39.80 & 357.71 & \textbf{42.25} & 357.96 & 40.97 & 351.91 & 41.07 & 349.95 \\
 & 20 & 10 & 39.56 & 358.14 & \textbf{42.06} & 358.34 & 40.90 & 352.05 & 41.01 & 349.79 \\
 & 20 & 20 & 39.53 & 357.21 & \textbf{42.04} & 358.28 & 40.90 & 351.09 & 41.07 & 349.76 \\
 \midrule
\multirow{8}{*}{\rot{\small\textit{California40}}} & 10 & -10 & 45.13 & 338.87 & 48.52 & 354.14 & 46.85 & 345.28 & \textbf{49.04} & 353.57 \\
 & 10 & 0 & 44.57 & 335.98 & 48.70 & 353.89 & 47.02 & 345.36 & \textbf{49.35} & 354.43 \\
 & 10 & 10 & 44.54 & 337.14 & 48.36 & 354.26 & 46.91 & 344.19 & \textbf{49.18} & 353.90 \\
 & 10 & 20 & 45.11 & 337.39 & 48.30 & 354.09 & 46.73 & 344.42 & \textbf{49.12} & 353.72 \\
 & 20 & -10 & 42.50 & 337.16 & 42.63 & 339.23 & 45.00 & 351.71 & \textbf{46.78} & 352.76 \\
 & 20 & 0 & 42.72 & 335.22 & 43.64 & 345.10 & 44.77 & 349.53 & \textbf{46.96} & 353.16 \\
 & 20 & 10 & 42.90 & 335.65 & 43.48 & 343.63 & 44.02 & 346.81 & \textbf{46.85} & 352.75 \\
 & 20 & 20 & 43.41 & 337.87 & 43.15 & 342.36 & 44.26 & 349.39 & \textbf{46.70} & 352.45 \\
\bottomrule
\end{tabular*}
\scriptsize
\parbox[b]{8.5cm}{The \textbf{bold} value indicates the best result.}
\end{threeparttable}
\end{table}

\subsubsection{Mission Success Rate} \label{sec:MSR}
Fig.~\ref{fig:prob-of-succ} presents the percentage of safe returns to the end depot (i.e., the UAV has more than the minimum allowed energy level) across 50 individual executions for each approach. Within every execution, all approaches utilize the same offline path. Our results show that ADAPT consistently achieves $100\%$ MSR across all test scenarios. This success may be attributed to its ability to identify high-quality paths with strong safety beliefs from the early stages of a mission. \textit{MCGreedy} also performs well on average because the most common path generally has good quality with certain robustness (as shown in Fig.~\ref{fig:path}). However, the random nature of MC approaches results in unstable performance, as evidenced in the \textit{California30} with low $\Delta\mu_{\mathrm{\bar{P}}^*}$. In contrast, \textit{ROMP} often fails in scenarios with high $\Delta\mu_{\mathrm{\bar{P}}^*}$, as it tends to overestimate the UAV's capacity, leading to delayed recognition of the mission failure risk. Although \textit{WeightedErr} incorporates online information for re-planning, its static weighting $w_\text{act} = w_\text{est} = 0.5$ proves inadequate to compensate for errors when $\Delta\mu_{\mathrm{\bar{P}}^*}$ is high. This highlights a practical challenge in determining global optimal values for $w_\text{act}$ and $w_\text{est}$, which requires extensive prior knowledge. Note that when solving \textit{California30} with $\Delta\mu_{\mathrm{\bar{P}}^*} = 0$, \textit{WeightedErr} exhibits a few failed paths, despite \textit{Offline} achieving a $100\%$ MSR. This occurs because the estimated path costs are close to the budget at each planning, leading to low error tolerance. The actual costs of these failed paths (360.37, 360.85, 359.72, and 359.82 kJ) all marginally exceed the budget constraint (359.64 kJ). 

\begin{figure*}[t]
    \hspace*{0.2cm}
    \centering
    \includegraphics[width=0.954\textwidth]{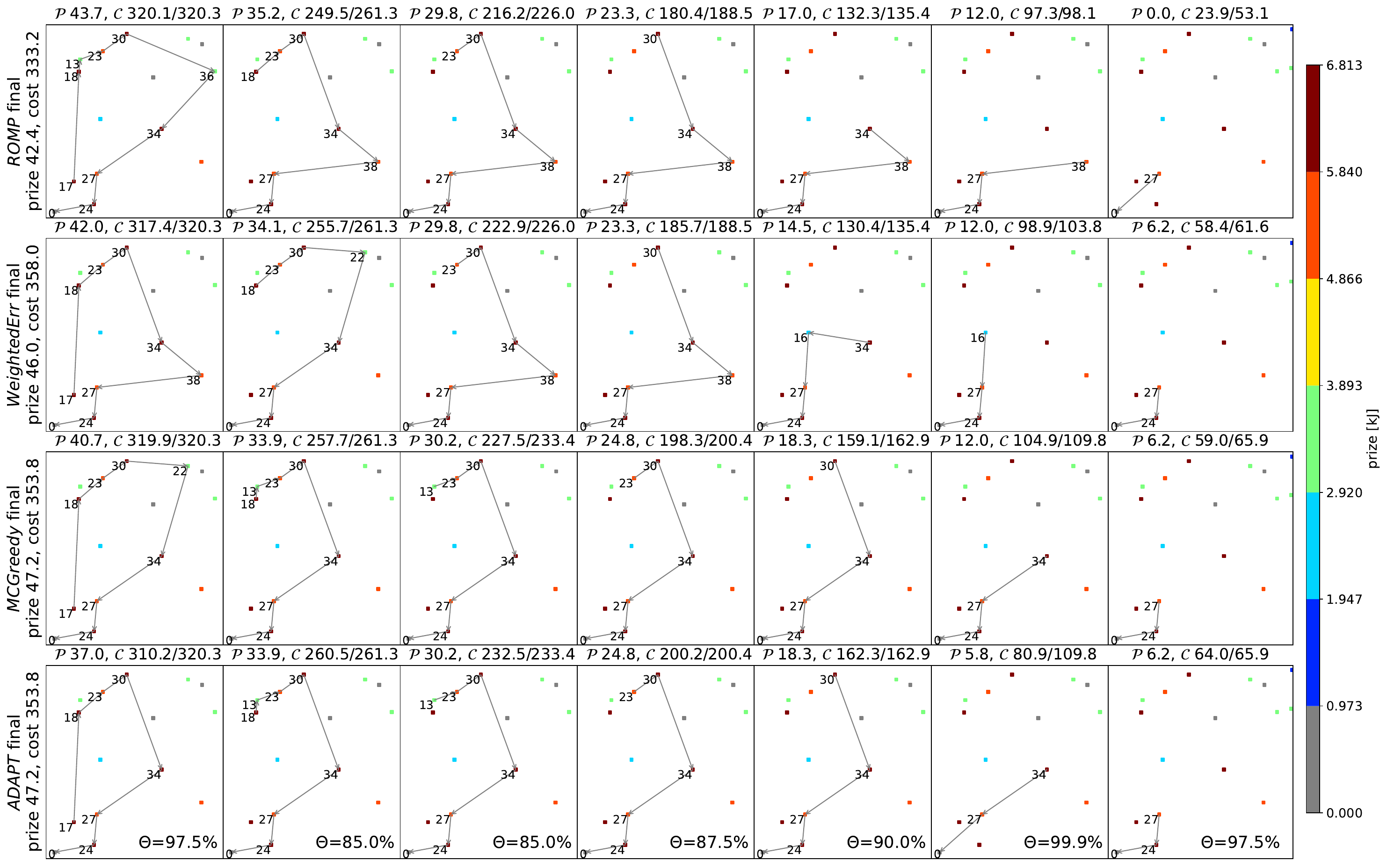}
    \caption{Online re-planning processes comparison for a typical execution of \textit{California40} with $\Delta\mu_{\mathrm{\bar{P}}^*}=20\%$ and $\Delta\sigma_{\mathrm{\bar{P}}^*}=-10\%$. Indices for sensor nodes not in the path are hidden. The value after `/' is the budget, and the initial budget is 359.64 kJ.}
    \label{fig:sol-analysis}
\end{figure*}

\subsubsection{Path prizes and costs}
Table~\ref{tab:sol-quality-short} presents the average solution quality for successful paths in solving \textit{California20} and \textit{California40} scenarios with high $\Delta\mu_{\mathrm{\bar{P}}^*}$ values (see Appendix F for full results). Note that SDs for these 50 executions are omitted because offline paths can have a weak effect on the final solution quality (as stated in Appendix C), and random sampling from distributions (especially for high $\Delta\sigma_{\mathrm{\bar{P}}^*}$) may introduce considerable uncertainty. ADAPT yields high-quality solutions across most scenarios compared to other approaches. In \textit{California20}, long distances between sensor nodes potentially lead to excessive energy costs for achieving high $\Theta$ values. Consequently, ADAPT adopts a more conservative strategy, as inserting or changing to new nodes becomes difficult. In the \textit{California40} scenario of Table~\ref{tab:sol-quality-short}, ADAPT tends to generate a high $\Theta$ path at early stages, subsequently improving it as more data is observed. 

\subsubsection{Mission process analysis}
Table~\ref{fig:sol-analysis} illustrates typical mission processes of all online approaches for \textit{California40} with $\Delta\mu_{\mathrm{\bar{P}}^*} = 20\%$ and $\Delta\sigma_{\mathrm{\bar{P}}^*} = -10\%$. \textit{MCGreedy} and ADAPT only insert node 13 during the mission (they have the same final path). In contrast, \textit{ROMP} and \textit{WeightedErr} underestimate actual power consumption, leading them to include nodes with high prizes (but deviating from the overall path) when the UAV has sufficient residual energy. Consequently, some nodes with low prizes but lower costs (or less risky) are dropped. For instance, in Fig.~\ref{fig:sol-analysis}, \textit{ROMP} drops nodes 13 and 36 to incorporate node 38 in the path. While the solutions of \textit{WeightedErr} and ADAPT suggest that dropping node 36 may be a viable strategy, visiting node 38 results in insufficient energy to visit and recharge node 24, ultimately leading to an inefficient solution. Similarly, \textit{WeightedErr} drops node 13 to incorporate node 38. Although it recognizes the risk of including node 38 when the UAV is at node 34 and attempts to compensate for the prize loss by switching to recharge node 16, the quality of the final solution is still compromised. 

\section{Conclusion and Future Work} \label{sec:conclusion}

This paper develops the UDOP, which aims to identify an optimal path that initiates from a start depot and returns to an end depot, maximizing the prize collection within the budget constraint under uncertain environments. UDOP differs from other OP variants with travel costs following distributions with unknown dynamic parameters and the potential impact of uncertain travel costs on node prizes and associated prize-collection costs. We propose a novel approach, ADAPT, to address the UDOP. In ADAPT, the offline planner generates an initial solution using prior knowledge of edge costs; the execution phase updates the mission execution status and records observations to online costs; the online planner employs a Bayesian approach to infer the parameters of edge cost distributions and determines the cost level (safety belief) for the solver, i.e., IACS. Because the re-planning happens at each node, the impact of uncertain edge costs is naturally involved in the optimization process. Our experimental results demonstrate that ADAPT can achieve a $100\%$ Mission Success Rate among all test instances. ADAPT can even yield solutions that outperform other benchmark approaches in some scenarios where the expected energy cost is much less than the actual.

We highlight several opportunities for additional research to explore UDOP further. The framework could be extended to incorporate uncertainties in prize-collection costs. For instance, in the CSP scenario, IPT link efficiency varies with coil alignment and medium properties. Moreover, ADAPT can be extended for various UAV types by pre-training a linear regression model using field flight data. Our future work will test ADAPT's generalization across various UAV types and conduct field experiments to verify its performance. We will also study the interplay between trajectory and mission planning. Examining how uncertainties are handled by the local controller (e.g., as presented in \cite{wang2024constrained}) can guide the global mission planner (outer loop control) to determine safety belief bounds and how safety belief affects trajectory tracking precision tolerance may advance autonomous system capabilities in Internet of Things contexts. Finally, extending UDOP to collaborative multi-UAV scenarios may require reformulation as a Team Orienteering Problem \cite{CHAO1996464} or Vehicle Routing Problem variant. This remains an open challenge that would require incorporating online information exchange and coordination.

\section*{Acknowledgments}
This work has been partially supported by the CHEDDAR: Communications Hub for Empowering Distributed ClouD Computing Applications and Research funded by the UK EPSRC under grant numbers EP/Y037421/1 and EP/X040518/1.

\bibliographystyle{IEEEtran}
\bibliography{ref}

\newpage

\appendices

\section{Proof of Theorem 1} \label{appdix:proof-theorem1}
\begin{theorem}
      The estimated average power consumption can be modeled as a normal distribution with mean $\mu_{\mathrm{\bar{P}}}=\mu(b_1) \sqrt{\frac{m^3}{\rho}} + \mu(b_0)$ and variance $\sigma^2_{\mathrm{\bar{P}}^*} = \sigma^2(b_1) \frac{m^3}{\rho} + \sigma^2(b_0)$.  
\end{theorem}
\begin{proof}
    Given $b_1$ and $b_0$ follow two independent normal distributions as indicated in Table~\ref{tab:reg-model-coef}, the estimated average power is the linear combination of their independent and identically distributed (i.i.d.) samples. This is mathematically equivalent to stating that the estimated average power follows a normal distribution with mean $\mu_{\mathrm{\bar{P}}}=\mu(b_1) \sqrt{\frac{m^3}{\rho}} + \mu(b_0)$ and variance $\sigma^2_{\mathrm{\bar{P}}^*} = \sigma^2(b_1) \frac{m^3}{\rho} + \sigma^2(b_0)$.
\end{proof}

\section{Inherited ant colony system} \label{appdix:iacs}
\begin{algorithm}[!b]
    \caption{Drop operator} \label{alg:drop-operator}
    \SetKwInOut{Input}{Input}%
    \Input{Path of ant $m$;\: Feasible node set $\mathbf{A}_m^{\mathbf{V}}$.}
    \While{ant path does \textbf{not} satisfy Constraint \eqref{con:sdop-budget}}{
        \For{each node $v_k$ at path index $l$ (exclude start and end node)}{
            Compute drop value $drop(v_k \:|\: l)$ by Eq. \eqref{eq:drop-value}\;
        }
        Find the path index $j$ at which the node $v_i$ has the minimum drop value, i.e., $i, j = \arg\min_{k, l} \big\{ drop (v_k \:|\: l), ... \big\}$\;
        Update path cost $\mathcal{C}_m \leftarrow \mathcal{C}_m - \mathcal{C}_{\text{drop}}(v_i\:|\: j)$\;
        Update path prize $\mathcal{P}_m \leftarrow \mathcal{P}_m - \mathcal{P}(v_{i})$\;
        Remove the node at path index $i$\;
        Update the feasible set $\mathbf{A}_m^{\mathbf{V}} \leftarrow \mathbf{A}_m^{\mathbf{V}} \cup \big\{ v_{i} \big\}$\;
    }
    \Return The feasible path with the new prize, new cost, and updated feasible node set.
\end{algorithm}

\begin{algorithm}[!b]
    \caption{Add operator} \label{alg:add-operator}
    \SetKwInOut{Input}{Input}%
    \Input{Ant path;\: Feasible node set $\mathbf{A}_m^{\mathbf{V}}$.} 
    \While{exist any node $\in \mathbf{A}_m^{\mathbf{V}}$ can be inserted into ant path without violating Constraint \eqref{con:sdop-budget}}{
        \For{each node $v_k \in \mathbf{A}_m^{\mathbf{V}}$}{
            Get 3 neighbor nodes in the ant path with minimum distance cost to visit $\mathbf{S}_{\text{nbr}} =\big\{ nbr_1,\: nbr_2,\: nbr_3 \big\}$\;
            \For{each pair of neighbor nodes $(nbr_i,\: nbr_j) \in \mathbf{S}_{\mathrm{nbr}}$}{
                Check if this pair is adjacent in the path\;
            }
            \If{\textbf{no} adjacent pair exists}{
                \For{$nbr_i \in \mathbf{S}_{\mathrm{nbr}}$}{
                    Find the previous and next node of $nbr_i$ in the path\;
                    Create new pair $(v_{prev},\: nbr_i)$ and $(nbr_i,\: v_{next})$\;
                }
            }
            Find the pair that minimizes the visitation cost and find at which index $l$ to insert\;
            Compute add value $add(v_k \:|\: l)$ by Eq. \eqref{eq:add-value}\;
        }
        Find the node $v_i$ with maximum add value and its insert index $j$ in the path, i.e., $i, j = \arg\max_{k, l} \big\{ add( v_k \:|\: l), ... \big\}$\;
        Update path cost $\mathcal{C}_m \leftarrow \mathcal{C}_m + \mathcal{C}_{\text{add}}(v_i \:|\: j)$\;
        Update path prize $\mathcal{P}_m \leftarrow \mathcal{P}_m + \mathcal{P}(v_i)$\;
        Insert the node $v_i$ into the ant path (at index $j$)\;
        Update the feasible set $\mathbf{A}_m^{\mathbf{V}} \leftarrow \mathbf{A}_m^{\mathbf{V}} \setminus \big\{ v_i \big\}$\;
    }
    \Return The feasible path with the new prize, new cost, and updated feasible node set.
\end{algorithm}

\begin{algorithm}[!b]
    \caption{Inherited Ant Colony System} \label{alg:iacs}
    \SetKwInOut{Input}{Input}%
    \Input{Node set $\mathbf{V}$; Number of ants $N_{\mathrm{ant}}$; Number of iterations $N_{\mathrm{it}}$; Maximum number of no improvement $N_{\mathrm{impr}}$; Improvement tolerance $\epsilon_{\mathrm{ACS}}$; $\beta$ in Eq. \eqref{eq:acs-random-prob}; $q_0$ in Eq. \eqref{eq:acs-state}; $\rho$ in Eq. \eqref{eq:acs-local}; $\alpha$ in Eq. \eqref{eq:acs-global}; Previous global-best path.}
    \If{online planning}{
        Apply Add operator (Alg. \ref{alg:add-operator}) and Drop operator (Alg. \ref{alg:drop-operator}) to the previous global-best path\;
        Update pheromone matrix with $\tau_0 \leftarrow \mathcal{P}_{\mathrm{gb}} / \mathcal{C}_{\mathrm{gb}}$\;
    }
    \Else{
        Update pheromone matrix with $\tau_0$ obtained by nearest neighbor heuristic\;
    }
    Initialize $N_{\mathrm{ant}}$ ants and associated feasible node set $\mathbf{A}^{\mathbf{V}}$, and set no improvement counter to $0$\;
    \For{$n_{\mathrm{it}}=1$ to $N_{\mathrm{it}}$}{
        \If{no improvement counter $\geq N_{\mathrm{impr}}$}{ Break the loop\; }
        \For{each ant $m$}{
            Randomly sample the first node $v \in \mathbf{A}_m^{\mathbf{V}}$ and add it to the ant path\;
            \While{$\mathbf{A}_m^{\mathbf{V}} \neq \varnothing$ \textbf{and} ant path satisfies Constraint \eqref{con:sdop-budget}}{
                Select the next node by Eq. \eqref{eq:acs-state} and add it to the ant path\;
                Update the prize, cost, and feasibility of the path\;
            }
            Add the end depot node, then update path cost and feasibility\;
            Apply 2-opt operator, then update path cost and feasibility\;
            \If{ant path \textbf{not} feasible}{
                Invoke the drop operator (Alg. \ref{alg:drop-operator})\;
            }
            Invoke the add operator (Alg. \ref{alg:add-operator})\;
            Update the pheromone matrix by Eq. \eqref{eq:acs-local}\;
        }
        Update the local-best ant with index equals to $\arg\max_k \big\{ \mathcal{P}_k \big\}$ (or $\arg\min_k \big\{ \mathcal{C}_k\big\}$ if $\mathcal{P}$ is maximum)\;
        \If{\big( $\mathcal{P}_{\mathrm{lb}} \geq \mathcal{P}_{\mathrm{gb}} + \epsilon_{\mathrm{ACS}}$ \big) \textbf{or} \big($\mathcal{P}_{\mathrm{lb}} = \mathcal{P}_{\mathrm{gb}}$ \textbf{and } $\mathcal{C}_{\mathrm{lb}} \leq \mathcal{C}_{\mathrm{gb}} - \epsilon_{\mathrm{ACS}}$ \big)}{
            Update the global-best ant\;
            Update the pheromone matrix by Eq. \eqref{eq:acs-global}\;
        }
        \Else{ No improvement counter $+1$\; }
        Reset all ants (except the global-best ant)\;
    }
    \Return The path sequence, path cost, and path prize of the global-best ant.
\end{algorithm}
\vspace*{-0.2cm}

The solver of online planner is adapted from the IACS in \cite{qian2024ceopn}. Compared to the classic ACS, IACS (see \autoref{alg:iacs}) utilizes the path from previous computation to initialize the pheromone matrix. Because ACS is initially designed for an unconstrained optimization scenario (i.e., Traveling Salesman Problem), we employ drop operator (see \autoref{alg:drop-operator}) and add operator (see \autoref{alg:add-operator}) to confine the budget constraint and maximize budget utilization. The drop cost of a node $v_i$ denotes the sum of visit cost and service cost (if the node to be dropped is at path index $j$). The drop value of $v_i$ is simply defined as its prize divided by its drop cost:
\begin{subequations}
\begin{flalign} \label{eq:drop-value}
    \nonumber
    &\mathcal{C}_{\text{drop}}(v_i \:|\: j) = - \mathcal{C}_{f_2}(v_{j-1},\: v_{j+1}) & \\
    \nonumber
    & \qquad\qquad + \mathcal{C}_{f_2}(v_{j-1},\: v_i) + \mathcal{C}_{f_2}(v_i,\: v_{j+1}) + \mathcal{C}_{f_3} (v_i)& \\
    & drop(v_i\:|\: j) = \mathcal{P}(v_i) \:/\: \mathcal{C}_{\text{drop}}(v_i \:|\: j) &
\end{flalign}
\end{subequations}
Similarly, a node $v_i$'s add value (if inserted at path index $j$) has the form as drop value:
\begin{subequations}
\begin{flalign} \label{eq:add-value}
    \nonumber
    &\mathcal{C}_{\text{add}}(v_i \:|\: j) = - \mathcal{C}_{f_2}(v_{j-1},\: v_{j+1}) & \\
    \nonumber
    & \qquad\qquad + \mathcal{C}_{f_2}(v_{j-1},\: v_i) + \mathcal{C}_{f_2}(v_i,\: v_{j+1}) + \mathcal{C}_{f_3} (v_i)& \\
    & add(v_i\:|\: j) = \mathcal{P}(v_i) \:/\: \mathcal{C}_{\text{add}}(v_i \:|\: j) &
\end{flalign}
\end{subequations}

ACS simulates the foraging behavior of an ant colony, incorporating three fundamental rules: the state transition rule, which decides the next visitation; the local updating rule, responsible for adjusting the pheromone trail visited by all ants; and the global updating rule, which updates the pheromone matrix based on the global-best ant. In our state transition rule, the probability for the ant $m$ at the node $v_r$ to visit the next node $v_s$ is defined as:
\begin{equation} \label{eq:acs-random-prob}
    \begin{aligned}
        &p_m(r, s) = &\\
        &\begin{dcases}
            \dfrac{\big[ \tau(r,\: s) \big] \cdot \big[ \eta(r,\: s) \big]^{\beta}}{\sum_{v_u \in\: \mathbf{A}_m^{\mathbf{V}}}\: \big[ \tau(r,\: u) \big] \cdot \big[ \eta(r,\: u) \big]^{\beta}}\;, & \text{if} ~ v_s \in \mathbf{A}_m^{\mathbf{V}}\\
            0\;, & \text{otherwise}
        \end{dcases}&
    \end{aligned}
\end{equation}
where $\tau(r, s)$ is the pheromone deposited on edge $e_{rs}$. We define the heuristic information $\eta(r, s) = \dfrac{\mathcal{P}(v_s)}{ \mathcal{C}_{f_2}(v_r,\: v_s) + \mathcal{C}_{d_3}(v_s)}$ as the ratio of the node $v_s$'s prize to the sum of edge cost between these two nodes and prize-collection cost. $\beta$ is a parameter to control the relative importance of pheromone versus heuristic information. We denote the feasible set of remaining nodes in the ant $m$ by $\mathbf{A}_m^{\mathbf{V}}$. To balance exploring and exploiting, the state transition rule introduces an additional parameter $q_0 \in \big( 0, 1 \big)$: 
\begin{equation} \label{eq:acs-state}
    s = \begin{dcases}
        \arg\max_{v_s \in\: \mathbf{A}_m^{\mathbf{V}}} \Big\{ \big[ \tau (r,\: s) \big] \cdot \big[ \eta(r, \: s) \big]^{\beta} \Big\}\; , & q \leq q_0\\
        s \sim p_k(r,\: s) ~\text{in Eq.~\eqref{eq:acs-random-prob}} \; , & q > q_0
    \end{dcases}
\end{equation}
The probability $q \in \mathbb R$ is randomly generated from a uniform distribution ranging in $[0, 1]$. Moreover, to reduce the probability of ants constructing the same solution, the local updating rule is applied to edges visited by ants after the solution construction phase:
\begin{eqnarray}\label{eq:acs-local}
    \tau(r, s) \leftarrow (1-\rho)\cdot \tau(r, s) + \rho \cdot \tau_0(r, s)
\end{eqnarray}
The evaporation rate $\rho \in (0, 1)$ is a constant that limits the accumulated pheromone on edge $e_{rs}$. In ACS, only the global-best ant, whose solution achieves the highest quality so far (i.e., either maximum prizes or minimum costs when prizes are the same), can deposit the pheromone at the end of each iteration. The global updating rule is defined as: 
\begin{eqnarray}\label{eq:acs-global}
    \tau(r,\: s) \leftarrow (1-\alpha)\cdot \tau(r,\: s) + \alpha \cdot \Delta \tau(r,\: s)
\end{eqnarray}
where $\alpha \in (0, 1)$ is a constant to control the pheromone decay rate, the deposited pheromone can be obtained by:
\begin{eqnarray} \label{eq:acs-delta-tau-ceop}
    \Delta\tau(r, s) = \begin{dcases} \mathcal{P}_{\text{gb}} \: / \: \mathcal{C}_{\text{gb}} \; , & \text{if $e_{rs} \in $ global-best path}\\
        0 \; , & \text{otherwise}
    \end{dcases}
\end{eqnarray}
$\mathcal{P}_{\text{gb}}$ and $\mathcal{C}_{\text{gb}}$ are the collected prize and cost of the global-best path, respectively. We opted for a straightforward 2-opt local search method for later path sequence improvement.

\begin{table*}[t]
\centering
\begin{threeparttable}
\caption{Offline solution performance comparison.} \label{tab:offline-path0}
\setlength\tabcolsep{0pt}
\begin{tabular*}{\textwidth}{@{\extracolsep{\fill}} cccccccc}
\toprule
\multirow{2}{*}{Instance} & \multirow{2}{*}{Alg.} & \multicolumn{2}{l}{Execution time (s)} & \multicolumn{2}{l}{Offline prize (kJ)} & \multicolumn{2}{l}{Offline cost (kJ)} \\ \cline{3-4} \cline{5-6} \cline{7-8} 
 &  & mean & SD & mean & SD & mean & SD \\
 \midrule
\multirow{2}{*}{\small\textit{California20}} & ACS & \textbf{4.381} & 0.283 & 47.051 & 0.000 & 358.196 & 0.366 \\
 & BnB & 12.191 & 1.547 & \textbf{47.052} & 0.004 & 358.875 & 0.495 \\
\midrule
\multirow{2}{*}{\small\textit{California30}} & ACS & \textbf{6.908} & 0.316 & \textbf{41.014} & 0.142 & 359.223 & 0.294 \\
& BnB & 301.118 & 0.040 & 40.965 & 0.196 & 359.225 & 0.477 \\
\midrule
\multirow{2}{*}{\small\textit{California40}} & ACS & \textbf{9.572} & 0.364 & 51.050 & 0.000 & 358.148 & 0.391 \\
& BnB & 226.609 & 53.026 & \textbf{51.055} & 0.012 & 358.920 & 0.577 \\
\bottomrule
\end{tabular*}
\end{threeparttable}
\end{table*}

\begin{table*}[t]
\centering
\begin{threeparttable}
\caption{Performance of solutions with different offline paths.} \label{tab:offline-final}
\setlength\tabcolsep{0pt}
\begin{tabular*}{\textwidth}{@{\extracolsep{\fill}} ccccccccccccc}
\toprule
\multirow{3}{*}{Instance} & \multirow{2}{*}{$\Delta\mu_{\bar{\mathbf{P}}^*}$} & \multirow{2}{*}{$\Delta\mu_{\bar{\mathbf{P}}^*}$} & \multicolumn{5}{c}{ADAPT with limited ACS offline path} & \multicolumn{5}{c}{ADAPT with BnB offline path} \\
\cline{4-8}\cline{9-13}
 &  &  & \multicolumn{1}{c}{MSR} & \multicolumn{2}{c}{$\mathcal{P}^*$(kJ)} & \multicolumn{2}{c}{$\mathcal{C}^*$(kJ)} & \multicolumn{1}{c}{MSR} & \multicolumn{2}{c}{$\mathcal{P}^*$(kJ)} & \multicolumn{2}{c}{$\mathcal{C}^*$(kJ)} \\
 \cline{5-6}\cline{7-8}\cline{10-11}\cline{12-13}
 & (\%) & (\%) & (\%) & mean & SD & mean & SD & (\%) & mean & SD & mean & SD \\
 \midrule
\multirow{4}{*}{\small\textit{California20}} & -10 & 0 & 100 & \textbf{50.542} & 0.941 & 351.345 & 3.482 & 100 & 50.395 & 0.620 & 350.891 & 2.832 \\
 & 0 & 0 & 100 & 47.435 & 0.600 & 351.807 & 3.724 & 100 & \textbf{47.627} & 0.538 & 349.735 & 4.223 \\
 & 10 & 0 & 100 & 44.425 & 0.906 & 352.015 & 3.979 & 100 & \textbf{44.455} & 0.654 & 349.989 & 2.556 \\
 & 20 & 0 & 100 & 41.196 & 0.887 & 349.896 & 5.178 & 100 & \textbf{41.296} & 0.757 & 350.419 & 5.523 \\
  \midrule
\multirow{4}{*}{\small\textit{California30}} & -10 & 0 & 100 & 45.324 & 0.977 & 352.144 & 2.684 & 100 & \textbf{45.636} & 0.375 & 349.538 & 3.273 \\
 & 0 & 0 & 100 & 41.393 & 0.713 & 352.542 & 4.129 & 100 & \textbf{41.570} & 0.656 & 351.489 & 3.682 \\
 & 10 & 0 & 100 & \textbf{37.736} & 1.012 & 353.390 & 3.461 & 100 & 37.434 & 0.884 & 353.613 & 2.997 \\
 & 20 & 0 & 100 & \textbf{34.381} & 1.434 & 348.555 & 4.315 & 100 & 33.876 & 1.642 & 345.609 & 5.194 \\
 \midrule
 \multirow{4}{*}{\small\textit{California40}} & -10 & 0 & 100 & 54.078 & 0.929 & 344.561 & 3.960 & 100 & \textbf{54.277} & 0.849 & 344.450 & 4.541 \\
 & 0 & 0 & 100 & 51.580 & 0.852 & 350.355 & 4.304 & 100 & \textbf{51.841} & 0.529 & 348.455 & 3.387 \\
 & 10 & 0 & 100 & 48.520 & 1.614 & 351.968 & 4.552 & 100 & \textbf{48.942} & 0.946 & 353.909 & 2.537 \\
 & 20 & 0 & 100 & 45.750 & 1.651 & 351.127 & 4.430 & 100 & \textbf{46.144} & 1.165 & 351.283 & 3.859 \\
 \bottomrule
\end{tabular*}
\end{threeparttable}
\end{table*}

\section{Offline planning with ACS} \label{appdix:offline}
We assess the Ant Colony System (ACS) algorithm's efficacy in solving \textit{California} instances, comparing it to an exact method (i.e., the Branch and Bound algorithm, BnB) implemented in Gurobi \cite{gurobi}. To avoid unnecessary computation, we impose a 5-minute execution time limit and a 0.01 minimum improvement tolerance. Our experiments focus on $\Delta\mu_{\bar{\mathbf{P}}}\in\{-10, 0, 10, 20\}\%$ and $\Delta\sigma_{\bar{\mathbf{P}}}=0\%$ for all \textit{California} instances. Table~\ref{tab:offline-path0} presents averaged results from 200 individual executions (50 per $\Delta\mu_{\bar{\mathbf{P}}}$ value). The data indicate that ACS achieves solutions comparable to BnB's while significantly reducing computation time. Notably, for the \textit{California30} instance, ACS outperforms BnB, likely due to the 5-minute time constraint being insufficient for Gurobi to identify a high-quality solution. While extended execution time might enable Gurobi to determine the optimal solution, such prolonged computation is impractical during mission execution.

Given the similar performance of ACS's and Gurobi's solutions, we validate the robustness of ADAPT by adjusting ACS parameters stated in Section \ref{sec:param-setting} to $N_{\text{Ant}} = 4$ $N_{\text{ACS}}=25$, $\epsilon_{\text{ACS}}=0.01$, and setting the number of no-improvement iterations to 5. This adjustment reduces ACS's performance in computing a lower-quality offline path. Table~\ref{tab:offline-final} compares the solution quality of ADAPT using offline paths computed by limited ACS and Gurobi over 50 executions. Our findings suggest that the offline path quality may have a weak effect on final solution quality. 

\section{Computation time} \label{appdix:computation-time}
 The theoretical worst-case computational complexity of ADAPT is $\mathcal{O}(N_{\text{SN}}^2)$, where $N_{\text{SN}}$ denotes the number of nodes because the 2-opt operator has a $\mathcal{O}(N_{\text{SN}}^2)$ complexity. However, as noted by \cite{qian2024ceopn}, the inheritance mechanism can advance ACS's convergence process in practice. Fig.~\ref{fig:comp-time} visualizes the typical computational time of four online approaches under a low and high problem complexity scenario\footnote[3]{All experiments were conducted on an Intel NUC11TNK with i7-11657G (2.8 GHz) CPU and 16 GB RAM.}. All approaches can complete computation within seconds, demonstrating ADAPT's potential for continuous real-time re-planning.
\begin{figure}[t]
    \hspace*{-0.3cm}
    \centering
    \includegraphics[width=0.48\textwidth]{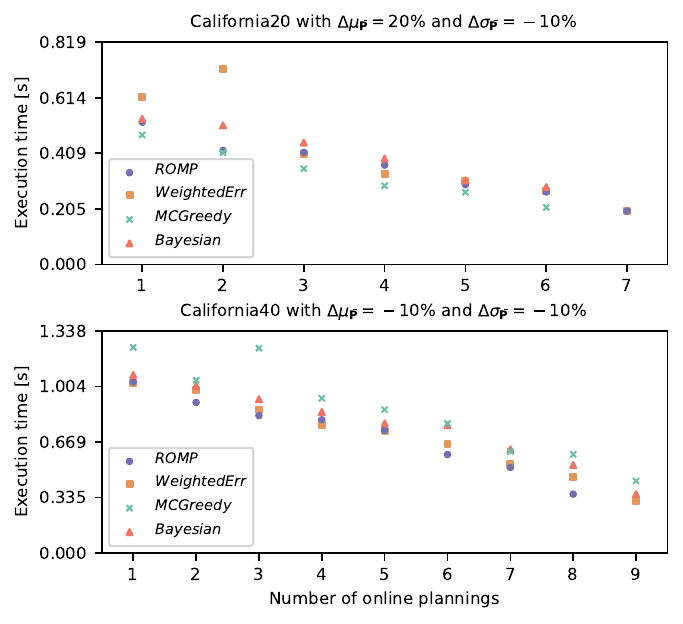}
    \caption{Two examples of algorithms' execution time. The top and bottom scenarios have the lowest and highest problem complexity among all tested scenarios.}
    \label{fig:comp-time}
\end{figure}

\section{Sensitivity analysis of the minimum safety belief} \label{appdix:theta-sensitivity}
Solutions generated by different $\Theta_{\min}$ under various scenarios are presented in Table~\ref{tab:sensitivity-belief}. In summary, all settings of $\Theta_{\min}$ can have a high mission success rate of over 50 executions under most scenarios. The setting of $\Theta\in[45, 99]\%$ can frequently find higher-quality solutions compared to others because it allows more search space. However, the balance between safety beliefs and prize collection is challenging to maintain, resulting in risky solutions that pursue high prize collection (see \textit{California20} with $\Delta\mu_{\bar{\mathbf{P}}} = 20\%$). In conclusion, the prize advancement achieved by setting low $\Theta_{\min}$ insufficiently compensates for mission safety within the CSP context. The adaptive setting of weights to the safety belief and prize collection may allow lower $\Theta_{\min}$ to achieve a higher mission success rate.

\begin{table*}[t]
\centering
\begin{adjustbox}{angle=90}
\begin{threeparttable}
\caption{Sensitivity analysis of safety belief $\Theta$.} \label{tab:sensitivity-belief}
\setlength\tabcolsep{0pt}
\begin{tabular*}{\textheight}{@{\extracolsep{\fill}} cccccccccccccccccc}
\toprule
 & $\Delta\mu_{\bar{\mathbf{P}}^*}$ & $\Delta\sigma_{\bar{\mathbf{P}}^*}$ & \multicolumn{3}{c}{$\Theta \in [45, 99]\%$} & \multicolumn{3}{c}{$\Theta \in [55, 99]\%$} & \multicolumn{3}{c}{$\Theta \in [65, 99]\%$} & \multicolumn{3}{c}{$\Theta \in [75, 99]\%$} & \multicolumn{3}{c}{$\Theta \in [85, 99]\%$} \\
 \cline{4-6} \cline{7-9} \cline{10-12} \cline{13-15} \cline{16-18}
 & (\%) & (\%) & \text{MSR}(\%) & $\mathcal{P}^*$(kJ) & $\mathcal{C}^*$(kJ) & \text{MSR}(\%) & $\mathcal{P}^*$(kJ) & $\mathcal{C}^*$(kJ) & \text{MSR}(\%) & $\mathcal{P}^*$(kJ) & $\mathcal{C}^*$(kJ) & \text{MSR}(\%) & $\mathcal{P}^*$(kJ) & $\mathcal{C}^*$(kJ) & \text{MSR}(\%) & $\mathcal{P}^*$(kJ) & $\mathcal{C}^*$(kJ) \\
 \midrule
\multirow{16}{*}{\rot{\normalsize\textit{California20}}} & -10 & -10 & 100 & \textbf{50.945} & 348.928 & 100 & 50.849 & 349.750 & 100 & 50.623 & 350.449 & 100 & 50.348 & 351.488 & 100 & 50.279 & 350.968 \\
 & -10 & 0 & 100 & 50.837 & 349.506 & 100 & \textbf{50.839} & 349.275 & 100 & 50.688 & 350.560 & 100 & 50.292 & 352.069 & 100 & 50.242 & 351.415 \\
 & -10 & 10 & 100 & \textbf{50.765} & 349.987 & 100 & 50.726 & 349.632 & 100 & 50.675 & 350.270 & 100 & 50.304 & 352.481 & 100 & 50.249 & 351.241 \\
 & -10 & 20 & 98 & \textbf{50.943} & 350.290 & 100 & 50.818 & 349.979 & 100 & 50.523 & 350.512 & 100 & 50.337 & 351.563 & 100 & 50.288 & 350.792 \\
 & 0 & -10 & 100 & 47.175 & 350.368 & 100 & 47.484 & 349.272 & 100 & \textbf{47.629} & 349.539 & 100 & 47.607 & 349.316 & 100 & 47.581 & 350.139 \\
 & 0 & 0 & 100 & 47.155 & 349.858 & 100 & 47.491 & 348.663 & 100 & \textbf{47.714} & 350.206 & 100 & 47.629 & 349.617 & 100 & 47.525 & 350.234 \\
 & 0 & 10 & 100 & 47.554 & 350.298 & 100 & \textbf{47.658} & 350.474 & 100 & 47.564 & 350.514 & 100 & 47.569 & 349.167 & 100 & 47.602 & 350.895 \\
 & 0 & 20 & 100 & 47.461 & 351.015 & 100 & 47.628 & 350.276 & 100 & 47.438 & 349.731 & 100 & \textbf{47.632} & 349.644 & 100 & 47.596 & 350.653 \\
 & 10 & -10 & 94 & \textbf{44.672} & 350.356 & 96 & 44.190 & 350.251 & 100 & 44.617 & 350.390 & 100 & 44.365 & 350.226 & 100 & 44.419 & 350.030 \\
 & 10 & 0 & 92 & \textbf{44.703} & 350.957 & 100 & 44.332 & 349.762 & 100 & 44.442 & 350.333 & 100 & 44.357 & 349.983 & 100 & 44.280 & 349.822 \\
 & 10 & 10 & 92 & \textbf{44.668} & 351.189 & 100 & 44.536 & 351.285 & 100 & 44.392 & 350.831 & 100 & 44.430 & 350.014 & 100 & 44.311 & 349.561 \\
 & 10 & 20 & 86 & \textbf{44.632} & 350.602 & 98 & 44.499 & 351.064 & 96 & 44.342 & 350.126 & 100 & 44.466 & 350.492 & 100 & 44.487 & 350.156 \\
 & 20 & -10 & 32 & 40.985 & 347.310 & 68 & 40.524 & 341.236 & 94 & 40.320 & 342.787 & 100 & 41.135 & 351.537 & 100 & \textbf{41.168} & 350.775 \\
 & 20 & 0 & 48 & 40.639 & 344.362 & 74 & 40.738 & 343.225 & 90 & 40.776 & 344.835 & 100 & \textbf{41.104} & 349.919 & 100 & 41.072 & 350.337 \\
 & 20 & 10 & 36 & 40.552 & 346.397 & 88 & 40.132 & 341.628 & 92 & 40.555 & 344.342 & 100 & 41.074 & 348.259 & 100 & \textbf{41.080} & 350.769 \\
 & 20 & 20 & 57 & 41.076 & 348.026 & 90 & 40.235 & 342.960 & 92 & 40.846 & 344.426 & 100 & \textbf{41.103} & 350.431 & 100 & 41.080 & 350.004 \\
 \midrule
\multirow{16}{*}{\rot{\normalsize\textit{California30}}} & -10 & -10 & 100 & 45.52 & 350.69 & 100 & \textbf{45.67} & 350.70 & 100 & 45.62 & 350.22 & 100 & 45.62 & 350.44 & 100 & 45.46 & 351.03 \\
 & -10 & 0 & 100 & 45.46 & 349.63 & 100 & \textbf{45.65} & 350.04 & 100 & 45.58 & 350.39 & 100 & 45.63 & 349.66 & 100 & 45.48 & 349.69 \\
 & -10 & 10 & 100 & \textbf{45.68} & 350.64 & 100 & 45.64 & 350.56 & 100 & 45.64 & 349.78 & 100 & 45.52 & 349.76 & 100 & 45.50 & 350.47 \\
 & -10 & 20 & 100 & \textbf{45.78} & 350.76 & 100 & 45.78 & 351.01 & 100 & 45.53 & 349.73 & 100 & 45.62 & 350.16 & 100 & 45.51 & 350.08 \\
 & 0 & -10 & 96 & 41.78 & 351.97 & 100 & 41.76 & 352.37 & 100 & \textbf{42.00} & 352.70 & 100 & 41.50 & 350.64 & 100 & 41.65 & 351.64 \\
 & 0 & 0 & 98 & 41.67 & 351.17 & 100 & 41.69 & 351.58 & 100 & \textbf{41.90} & 352.67 & 100 & 41.80 & 352.25 & 100 & 41.56 & 350.78 \\
 & 0 & 10 & 100 & \textbf{41.72} & 351.27 & 98 & 41.66 & 352.31 & 100 & 41.55 & 350.73 & 100 & 41.57 & 351.38 & 100 & 41.54 & 350.60 \\
 & 0 & 20 & 100 & \textbf{41.65} & 351.73 & 100 & 41.55 & 351.15 & 100 & 41.49 & 350.45 & 100 & 41.57 & 351.25 & 100 & 41.48 & 350.48 \\
 & 10 & -10 & 98 & 38.31 & 354.59 & 100 & \textbf{38.50} & 355.19 & 100 & 38.26 & 354.21 & 100 & 37.66 & 353.27 & 100 & 37.39 & 352.62 \\
 & 10 & 0 & 96 & 38.39 & 355.26 & 100 & \textbf{38.52} & 354.78 & 100 & 38.05 & 354.22 & 100 & 37.61 & 353.48 & 100 & 37.63 & 352.73 \\
 & 10 & 10 & 100 & \textbf{38.48} & 355.07 & 100 & 38.37 & 355.02 & 100 & 38.31 & 354.99 & 100 & 37.66 & 353.25 & 100 & 37.53 & 353.05 \\
 & 10 & 20 & 100 & \textbf{38.44} & 355.31 & 100 & 38.12 & 354.68 & 100 & 38.19 & 354.88 & 100 & 38.02 & 354.00 & 100 & 37.18 & 351.89 \\
 & 20 & -10 & 96 & \textbf{34.64} & 352.66 & 94 & 34.56 & 351.70 & 100 & 34.42 & 349.67 & 100 & 34.07 & 346.50 & 100 & 33.69 & 345.08 \\
 & 20 & 0 & 100 & \textbf{34.73} & 352.76 & 100 & 34.47 & 352.58 & 96 & 34.41 & 348.69 & 100 & 34.14 & 346.53 & 100 & 33.96 & 346.85 \\
 & 20 & 10 & 100 & \textbf{34.79} & 352.44 & 96 & 34.30 & 351.12 & 96 & 34.27 & 350.31 & 100 & 33.85 & 345.24 & 100 & 34.04 & 346.40 \\
 & 20 & 20 & 98 & 34.53 & 352.74 & 92 & 34.41 & 351.61 & 100 & 34.41 & 349.92 & 100 & 33.68 & 346.71 & 100 & \textbf{34.53} & 346.61 \\
 \midrule
\multirow{16}{*}{\rot{\normalsize\textit{California40}}} & -10 & -10 & 100 & 53.89 & 344.31 & 100 & 53.88 & 342.00 & 100 & 53.89 & 342.04 & 100 & \textbf{54.28} & 343.97 & 100 & 53.89 & 341.97 \\
 & -10 & 0 & 100 & 53.83 & 345.23 & 100 & 53.91 & 342.04 & 100 & 53.86 & 341.94 & 100 & \textbf{54.14} & 343.18 & 100 & 53.93 & 342.06 \\
 & -10 & 10 & 100 & 53.93 & 343.50 & 100 & 53.89 & 342.24 & 100 & 54.00 & 342.68 & 100 & \textbf{54.09} & 342.94 & 100 & 53.92 & 342.07 \\
 & -10 & 20 & 100 & 53.83 & 343.90 & 100 & 53.84 & 341.69 & 100 & 53.91 & 342.02 & 100 & \textbf{54.19} & 343.27 & 100 & 53.85 & 341.92 \\
 & 0 & -10 & 100 & 51.72 & 353.04 & 100 & \textbf{51.77} & 352.32 & 100 & 51.68 & 352.47 & 100 & 51.60 & 349.82 & 100 & 51.61 & 346.46 \\
 & 0 & 0 & 100 & 51.74 & 351.80 & 100 & \textbf{51.77} & 352.44 & 100 & 51.69 & 352.80 & 100 & 51.71 & 351.09 & 100 & 51.58 & 346.06 \\
 & 0 & 10 & 100 & 51.57 & 353.03 & 100 & \textbf{51.77} & 351.51 & 100 & 51.76 & 352.17 & 100 & 51.76 & 349.45 & 100 & 51.59 & 347.14 \\
 & 0 & 20 & 100 & 51.76 & 353.10 & 100 & 51.60 & 353.36 & 100 & 51.66 & 351.42 & 100 & \textbf{51.78} & 351.29 & 100 & 51.58 & 346.51 \\
 & 10 & -10 & 100 & 49.72 & 355.49 & 100 & 49.69 & 355.45 & 100 & \textbf{49.72} & 355.49 & 100 & 49.19 & 354.08 & 100 & 48.84 & 352.95 \\
 & 10 & 0 & 100 & \textbf{49.72} & 355.44 & 100 & 49.69 & 355.44 & 100 & 49.72 & 355.56 & 100 & 49.17 & 354.04 & 100 & 48.50 & 351.94 \\
 & 10 & 10 & 100 & \textbf{49.72} & 355.46 & 100 & 49.72 & 355.44 & 100 & 49.72 & 355.45 & 100 & 49.15 & 353.81 & 100 & 48.73 & 352.68 \\
 & 10 & 20 & 100 & \textbf{49.72} & 355.51 & 100 & 49.72 & 355.56 & 100 & 49.72 & 355.47 & 100 & 49.18 & 353.89 & 100 & 48.73 & 352.65 \\
 & 20 & -10 & 100 & \textbf{47.23} & 353.80 & 100 & 47.13 & 353.59 & 100 & 47.04 & 353.31 & 100 & 46.76 & 352.60 & 100 & 45.43 & 349.14 \\
 & 20 & 0 & 100 & \textbf{47.23} & 353.81 & 100 & 47.13 & 353.63 & 100 & 47.10 & 353.48 & 100 & 46.91 & 353.00 & 100 & 45.30 & 348.62 \\
 & 20 & 10 & 100 & \textbf{47.23} & 353.76 & 100 & 47.11 & 353.56 & 100 & 47.15 & 353.61 & 100 & 46.75 & 352.63 & 100 & 45.15 & 348.42 \\
 & 20 & 20 & 100 & \textbf{47.23} & 353.71 & 100 & 47.16 & 353.73 & 100 & 47.13 & 353.55 & 100 & 46.70 & 352.39 & 100 & 45.43 & 349.12 \\
 \bottomrule
\end{tabular*}
\end{threeparttable}
\end{adjustbox}
\end{table*}

\section{Full result of all tests} \label{appdix:full-result}
Table~\ref{tab:sol-quality-full} presents full results of \textit{Offline}, \textit{ROMP}, \textit{WeightedErr}, \textit{MCGreedy}, and ADAPT for solving \textit{California20}, \textit{California30} and \textit{California40} with $\Delta\mu_{\bar{\mathbf{P}}}, \Delta\sigma_{\bar{\mathbf{P}}}\in\{-10, 0, 10, 20\}\%$. These results demonstrate the average prizes and costs of successful paths over 50 executions.  
\begin{table*}[b]
\centering
\begin{threeparttable}
\caption{Solution quality comparison.} \label{tab:sol-quality-full}
\setlength\tabcolsep{0pt}
\begin{tabular*}{\textwidth}{@{\extracolsep{\fill}} ccccccccccccc}
\toprule
& $\Delta\mu_{\bar{\mathbf{P}}^*}$ & $\Delta\sigma_{\bar{\mathbf{P}}^*}$ & \multicolumn{2}{c}{\textit{Offline}} & \multicolumn{2}{c}{\textit{ROMP}} & \multicolumn{2}{c}{\textit{WeightedErr}} & \multicolumn{2}{c}{\textit{MCGreedy}} & \multicolumn{2}{c}{\textit{Bayesian}} \\ 
\cline{4-5} \cline{6-7} \cline{8-9} \cline{10-11} \cline{12-13}
& (\%) & (\%) & $\mathcal{P}^*$(kJ) & $\mathcal{C}^*$(kJ) & $\mathcal{P}^*$(kJ) & $\mathcal{C}^*$(kJ) & $\mathcal{P}^*$(kJ) & $\mathcal{C}^*$(kJ) & $\mathcal{P}^*$(kJ) & $\mathcal{C}^*$(kJ) & $\mathcal{P}^*$(kJ) & $\mathcal{C}^*$(kJ) \\ \midrule
\multirow{16}{*}{\normalsize\rot{\textit{California20}}} & -10 & -10 & 47.073 & 322.476 & 48.706 & 341.331 & \textbf{51.033} & 348.916 & 49.662 & 351.174 & 50.484 & 351.182 \\
 & -10 & 0 & 47.073 & 322.438 & 48.766 & 342.540 & \textbf{50.793} & 348.868 & 50.085 & 351.163 & 50.273 & 352.389 \\
 & -10 & 10 & 47.073 & 322.524 & 48.747 & 341.920 & \textbf{50.993} & 348.541 & 49.427 & 351.138 & 50.422 & 351.075 \\
 & -10 & 20 & 47.073 & 322.435 & 48.771 & 343.212 & \textbf{50.690} & 348.179 & 49.577 & 351.010 & 50.293 & 352.291 \\
 & 0 & -10 & 47.073 & 345.153 & 47.275 & 346.929 & \textbf{47.819} & 351.914 & 46.498 & 350.979 & 47.534 & 348.960 \\
 & 0 & 0 & 47.073 & 345.286 & 47.380 & 348.242 & \textbf{47.949} & 352.862 & 46.139 & 351.477 & 47.619 & 349.715 \\
 & 0 & 10 & 47.073 & 345.042 & 47.440 & 348.208 & \textbf{47.952} & 352.322 & 46.411 & 351.727 & 47.607 & 349.237 \\
 & 0 & 20 & 47.073 & 345.220 & 47.312 & 347.372 & \textbf{47.863} & 351.994 & 45.500 & 351.133 & 47.608 & 349.366 \\
 & 10 & -10 & N/A & N/A & \textbf{45.454} & 357.707 & 44.555 & 355.272 & 43.206 & 351.034 & 44.570 & 349.935 \\
 & 10 & 0 & N/A & N/A & \textbf{46.219} & 359.274 & 44.483 & 355.431 & 43.150 & 351.206 & 44.474 & 350.054 \\
 & 10 & 10 & N/A & N/A & \textbf{45.836} & 358.480 & 43.957 & 354.046 & 42.535 & 350.425 & 44.400 & 349.981 \\
 & 10 & 20 & N/A & N/A & \textbf{45.709} & 357.445 & 45.017 & 356.905 & 42.802 & 350.408 & 44.724 & 351.774 \\
 & 20 & -10 & N/A & N/A & 39.570 & 357.186 & \textbf{42.049} & 357.937 & 40.932 & 351.680 & 40.977 & 349.407 \\
 & 20 & 0 & N/A & N/A & 39.796 & 357.707 & \textbf{42.247} & 357.965 & 40.971 & 351.914 & 41.072 & 349.949 \\
 & 20 & 10 & N/A & N/A & 39.557 & 358.141 & \textbf{42.061} & 358.341 & 40.896 & 352.045 & 41.009 & 349.785 \\
 & 20 & 20 & N/A & N/A & 39.525 & 357.212 & \textbf{42.044} & 358.282 & 40.902 & 351.093 & 41.073 & 349.764 \\
\midrule
\multirow{16}{*}{\normalsize\rot{\textit{California30}}} & -10 & -10 & 41.075 & 323.247 & 44.335 & 347.762 & 45.521 & 349.337 & \textbf{45.895} & 353.571 & 45.641 & 349.310 \\
 & -10 & 0 & 41.053 & 323.023 & 44.204 & 347.368 & 45.556 & 348.978 & \textbf{46.075} & 352.896 & 45.558 & 349.610 \\
 & -10 & 10 & 41.041 & 323.202 & 44.271 & 347.525 & 45.532 & 348.957 & 45.582 & 353.406 & \textbf{45.610} & 349.761 \\
 & -10 & 20 & 41.043 & 323.255 & 44.010 & 346.753 & 45.624 & 349.661 & 45.478 & 352.451 & \textbf{45.636} & 350.437 \\
 & 0 & -10 & 41.045 & 347.544 & 41.931 & 353.072 & \textbf{42.280} & 355.925 & 41.187 & 352.238 & 41.923 & 352.405 \\
 & 0 & 0 & 41.048 & 347.660 & 41.846 & 352.650 & \textbf{42.495} & 355.819 & 41.164 & 352.275 & 41.674 & 351.292 \\
 & 0 & 10 & 41.053 & 347.632 & 41.867 & 352.676 & \textbf{42.510} & 356.231 & 40.900 & 351.290 & 41.763 & 351.944 \\
 & 0 & 20 & 41.049 & 347.730 & 41.904 & 353.052 & \textbf{42.246} & 355.701 & 40.593 & 349.987 & 41.595 & 350.996 \\
 & 10 & -10 & N/A & N/A & 37.806 & 357.019 & \textbf{38.119} & 356.690 & 37.067 & 352.241 & 37.663 & 352.928 \\
 & 10 & 0 & N/A & N/A & 37.882 & 356.787 & \textbf{38.218} & 356.949 & 36.890 & 351.574 & 37.347 & 353.057 \\
 & 10 & 10 & N/A & N/A & 37.907 & 356.441 & \textbf{37.946} & 356.143 & 37.098 & 351.455 & 37.612 & 352.927 \\
 & 10 & 20 & N/A & N/A & \textbf{37.989} & 357.361 & 37.894 & 356.697 & 37.024 & 351.947 & 37.629 & 353.548 \\
 & 20 & -10 & N/A & N/A & 34.009 & 352.862 & \textbf{34.442} & 358.718 & 34.293 & 351.754 & 34.241 & 347.321 \\
 & 20 & 0 & N/A & N/A & 34.227 & 353.285 & \textbf{34.344} & 357.202 & 34.170 & 350.692 & 33.709 & 344.794 \\
 & 20 & 10 & N/A & N/A & 34.341 & 353.621 & \textbf{34.442} & 358.631 & 34.297 & 349.624 & 33.882 & 346.414 \\
 & 20 & 20 & N/A & N/A & 33.994 & 352.817 & \textbf{34.519} & 359.040 & 34.195 & 350.364 & 34.346 & 347.228 \\
\midrule
\multirow{16}{*}{\normalsize\rot{\textit{California40}}} & -10 & -10 & 51.076 & 321.772 & 52.366 & 330.673 & \textbf{54.419} & 345.336 & 54.228 & 349.600 & 54.255 & 343.772 \\
 & -10 & 0 & 51.076 & 321.747 & 52.366 & 330.855 & \textbf{54.360} & 345.074 & 54.294 & 348.465 & 54.215 & 343.539 \\
 & -10 & 10 & 51.076 & 321.802 & 52.367 & 330.789 & \textbf{54.287} & 345.141 & 53.754 & 347.906 & 54.113 & 343.101 \\
 & -10 & 20 & 51.076 & 321.734 & 52.396 & 331.243 & \textbf{54.363} & 345.175 & 53.207 & 347.779 & 54.192 & 343.457 \\
 & 0 & -10 & 51.076 & 343.318 & 51.263 & 347.121 & \textbf{52.333} & 351.466 & 50.190 & 347.960 & 51.696 & 351.336 \\
 & 0 & 0 & 51.076 & 343.387 & 51.262 & 347.091 & \textbf{52.333} & 351.485 & 50.599 & 348.009 & 51.540 & 349.952 \\
 & 0 & 10 & 51.076 & 343.493 & 51.265 & 347.157 & \textbf{52.333} & 351.437 & 50.776 & 348.726 & 51.617 & 350.529 \\
 & 0 & 20 & 51.076 & 343.410 & 51.283 & 347.128 & \textbf{52.319} & 351.064 & 50.514 & 347.348 & 51.755 & 350.744 \\
 & 10 & -10 & N/A & N/A & 45.127 & 338.869 & 48.515 & 354.141 & 46.850 & 345.281 & \textbf{49.042} & 353.574 \\
 & 10 & 0 & N/A & N/A & 44.573 & 335.979 & 48.697 & 353.888 & 47.020 & 345.358 & \textbf{49.347} & 354.428 \\
 & 10 & 10 & N/A & N/A & 44.537 & 337.141 & 48.359 & 354.259 & 46.905 & 344.185 & \textbf{49.183} & 353.898 \\
 & 10 & 20 & N/A & N/A & 45.111 & 337.387 & 48.302 & 354.085 & 46.732 & 344.422 & \textbf{49.124} & 353.721 \\
 & 20 & -10 & N/A & N/A & 42.503 & 337.158 & 42.629 & 339.227 & 44.998 & 351.714 & \textbf{46.779} & 352.763 \\
 & 20 & 0 & N/A & N/A & 42.720 & 335.224 & 43.642 & 345.100 & 44.765 & 349.531 & \textbf{46.957} & 353.161 \\
 & 20 & 10 & N/A & N/A & 42.903 & 335.647 & 43.484 & 343.633 & 44.018 & 346.811 & \textbf{46.850} & 352.747 \\
 & 20 & 20 & N/A & N/A & 43.414 & 337.874 & 43.153 & 342.360 & 44.263 & 349.385 & \textbf{46.696} & 352.454 \\
\bottomrule
\end{tabular*}
\end{threeparttable}
\end{table*}

\vfill

\end{document}